\documentclass[reprint,amsmath,amssymb,aps]{revtex4-2}

\usepackage[colorlinks=true,allcolors=blue,]{hyperref}

\usepackage{amsfonts,amsthm}
\usepackage{array}
\usepackage{dcolumn}
\usepackage{xcolor,soul} 
\usepackage{booktabs}
\usepackage{adjustbox}
\usepackage{physics}
\usepackage{caption}
\usepackage{subcaption}
\newcommand{\ra}[1]{\renewcommand{\arraystretch}{#1}}
\usepackage{tikz}
\usetikzlibrary{calc,arrows,positioning,backgrounds,fit, quantikz, cd}

\newcommand{\intensityent}{50}
\newcommand{\intensity}{50} 
\newcommand{\intensityb}{25}

\definecolor{cbox}{HTML}{0072B2}
\definecolor{cenc}{HTML}{D55E00}
\definecolor{cvar}{HTML}{009E73}

\theoremstyle{plain}
\newtheorem{thm}{Theorem}[section]
\newtheorem{lem}[thm]{Lemma}
\newtheorem{prop}[thm]{Proposition}
\newtheorem{cor}[thm]{Corollary}
\theoremstyle{definition}
\newtheorem{defn}{Definition}[section]
\newtheorem{example}{Example}

\newcommand{\defeq}{\mathrel{:\mkern-0.25mu=}}
\newcommand{\CZ}{\operatorname{CZ}}
\newcommand{\CNOT}{\operatorname{CNOT}}
\newcommand{\CRX}{\operatorname{CRX}}
\newcommand{\CAN}{\operatorname{CAN}}
\newcommand{\B}{\fontseries{b}\selectfont}
\newcommand{\cc}[1]{\overline{#1}}

\begin{document}
	
	\title{Learning Fourier series with parametrized quantum circuits}
	
	\author{Dirk Heimann}
	\email{dirk.heimann@uni-bremen.de}
	\affiliation{Robotics Research Group, University of Bremen, 28359 Bremen, Germany}
	\author{Gunnar Schönhoff}
	\email{gunnar.schoenhoff@dfki.de}
	\affiliation{German Research Center for Artificial Intelligence - Robotics Innovation Center (DFKI RIC), 28359 Bremen, Germany}
	\author{Elie Mounzer}
	\email{elie.mounzer@dfki.de}
	\affiliation{German Research Center for Artificial Intelligence - Robotics Innovation Center (DFKI RIC), 28359 Bremen, Germany}
	\author{Hans Hohenfeld}
	\email{hans.hohenfeld@uni-bremen.de}
	\affiliation{Robotics Research Group, University of Bremen, 28359 Bremen, Germany}
	\author{Frank Kirchner}
	\email{frank.kirchner@dfki.de}
	\affiliation{Robotics Research Group, University of Bremen, 28359 Bremen, Germany}
	\affiliation{German Research Center for Artificial Intelligence - Robotics Innovation Center (DFKI RIC), 28359 Bremen, Germany}
	
	\date{\today}
	
	\begin{abstract}
		Variational quantum algorithms (VQAs) and their applications in the field of quantum machine learning through parametrized quantum circuits (PQCs) are thought to be one major way of leveraging noisy intermediate-scale quantum computing devices.
		However, differences in the performance of certain VQA architectures are often unclear since established best practices, as well as detailed studies, are missing.
		In this paper, we build upon the work by Schuld et al.~[Phys. Rev. A 103, 032430 (2021)] and Vidal et al.~[Front. Phys. 8, 297 (2020)]
		by comparing how well popular ans\"atze for PQCs learn different one-dimensional truncated Fourier series.
		We also examine dissipative quantum neural networks (dQNN) as introduced by Beer et al.~[Nat. Commun. 11, 808 (2020)] and propose a data reupload structure for dQNNs to increase their
		capability for this regression task.
		By comparing the results for different PQC architectures,
		we can provide guidelines for designing efficient PQCs.
	\end{abstract}
	
	\maketitle
	
	\section{Introduction}
	\label{Introduction}
	
	Quantum machine learning (QML) is one of the areas of quantum computing that has attracted a lot of attention in recent years~\cite{Biamonte2017, Mitarai2018, Schuld2019, Havlicek2019, Dunjko2020, SchuldBuch2021, Cerezo2022}, especially because QML algorithms seem to be a promising candidate for making use of noisy intermediate-scale quantum computing (NISQ)~\cite{Preskill2018, Bharti2022}.
	There are already a variety of prototypical practical applications of QML, ranging from the classification of radiological~\cite{Matic2022} or car images~\cite{Sagingalieva2022} to other industrial use cases~\cite{Mangini2022}, reinforcement learning environments~\cite{Skolik2022} and robot pathfinding~\cite{Heimann2022}.
	Among the methods investigated for QML are variational quantum algorithms (VQAs) which make use of parametrized quantum circuits (PQCs) in a hybrid quantum-classical optimization loop~\cite{Cerezo2021VQAs}.
	The term \textit{quantum neural networks} is often used synonymously with this application of PQCs for QML tasks~\cite{Schuld2014, Abbas2021}.
	
	While improving the performance and scalability of VQAs is an area of ongoing research (see, for example, Refs.~\cite{Caro2022, Larocca2022, Meyer2022}), and the question of whether QML algorithms can eventually outperform classical algorithms is still open
	(see, for example, Refs.~\cite{Cerezo2021VQAs, Schreiber2022, gujju2024}),
		significant steps have already been made~\cite{Liu2021, Jaeger2022, jerbi2024}.
	For classification tasks on classical data, Ref.~\cite{bowles2024}
	shows empirically that evidence for hybrid QML algorithms outperforming
	classical neural networks is still missing.
	Different circuit architectures or \textit{ans\"atze} for PQCs have been proposed and used, ranging from hardware-efficient circuits~\cite{Kandala2017} and the quantum alternating operator ansatz (QAOA)~\cite{Hadfield2019} to dissipative quantum neural networks (dQNNs)~\cite{Beer2020}, which were successfully implemented on a six-qubit superconducting processor~\cite{pan2023deep}.
	Additionally, data reupload has been shown to improve VQA performance significantly when encoding classical input data~\cite{PerezSalinas2020, Vidal2020, Schuld2021}.
	
	Steps to compare different PQCs were done by several groups, for example
	in Ref.~\cite{Schuld2021}, where a methodology based on Fourier series and coefficients was developed to compare the expressiveness of PQCs;
	in Ref.~\cite{li2022}, where architectures with amplitude-encoding and varying circuit depths were compared on different benchmark classification tasks;
	in Refs.~\cite{Sukin2019},~\cite{Rasmussen2020}, and~\cite{Hubregtsen2021}, where selected architectures were evaluated numerically regarding expressibility and entangling capability; in Ref.~\cite{Ballarin2022}, where the entanglement and the entangling speed of different circuits were compared; in Ref.~\cite{Wilkinson2022}, where effective capacity and effective dimension were analyzed for different types of PQCs, including dQNNs; and in Ref.~\cite{Yu2022}, where the work from Ref.~\cite{Schuld2021} was investigated in more detail for single-qubit circuits.
	Furthermore, the authors of Ref.~\cite{Schreiber2022} investigate the performance of classical surrogates for layered PQCs.
	While this fact indicates that it may be difficult to find advantages of PQCs
		in comparison to classical machine learning on classical data,
		this question is not settled and, for example,	
		the authors of Ref.~\cite{jerbi2024} show that so-called flipped quantum models exist,
		which have a provable learning advantage over fully classical learners
		despite having a classical surrogate model.
	Also, as laid out in, for example, Refs.~\cite{arunachalam2017, Caro2022_phd}, mathematically rigorous bounds exist that link model complexity of PQCs with their ability to generalize, providing a starting point for quantum learning theory.
	
	Despite all this work, general established best practices for the use of PQC architectures are still missing and the construction of efficient, trainable PQCs is an ongoing research topic (see, for example, Refs.~\cite{Funcke2021, Katabarwa2022}).
	Existing methods for performance evaluation of PQC ans\"atze are not enough to fully describe the ansatz performance since they do not take into account the combination of expressiveness and trainability that is needed in practical use cases.
	Furthermore, as we show in Sec.~\ref{subsec:schuld} only evaluating the Fourier coefficients resulting from the inverse Fourier transform of an ansatz evaluation with sampled parameters also misses differences between certain PQC architectures.
	In this paper, we tackle this problem by extending and using the numerical methodology presented in Ref.~\cite{Schuld2021} to lay out a framework to both more strictly and practically compare the performance of small circuit elements as well as complete PQC ans\"atze.
	
	In particular, our contributions are the following:
	\begin{itemize}
		\item We extend the numerical results by Schuld et al.~\cite{Schuld2021}, in particular by showing that PQCs with the same sampled Fourier coefficients can have different training results.
		\item We propose the new measure of \textit{learning capability} to analyze PQC ansatz performance.
		\item We compare different types of ans\"atze regarding this measure.
		\item We define and evaluate a data reupload structure for dQNNs.
	\end{itemize}
	
	In addition to other challenges~\cite{Anschuetz2022}, barren plateaus remain a major issue for the practical use of VQAs~\cite{McClean2018, Cerezo2021Cost, Marrero2021, Holmes2021, Holmes2022, sharma2022, Larocca2022barren}. For example, in Ref.~\cite{Larocca2022barren}, the authors show that PQCs with full rank dynamical Lie algebras exhibit barren plateaus because, as the number of layers increases, they become $\epsilon$-approximate 2-designs.
	In this paper, we compare architectures which all tend to exhibit barren plateaus for larger depth or global cost functions
	as shown in Appendix~\ref{ap:lie_algebra}, especially Corollary.~\ref{cor:lie_algebra}.
	The architectures utilize local cost functions and can exhibit barren plateaus if the depth is of polynomial order in the number of qubits~\cite{Cerezo2021Cost}.
	Our analysis gives guidelines to minimize layers, qubits, gates, and parameters 
	which can be used to reduce the depth.
	Thus, our work helps to reduce the effect of barren plateaus.

	These guidelines are drawn from numerically investigating how well architectures can learn Fourier series, which can, in principle,
	approximate arbitrary functions up to a certain precision. 
	Hence, we conclude that the
	guidelines for PQCs obtained from these results
	can be transferred to more general learning tasks.
	However, this assumption may not hold universally.
	
	The sections of this paper are organized as follows:
	In Sec.~\ref{PriorWork}, we summarize some of the existing work on PQC evaluation.
	In Sec.~\ref{Methods}, we present our approach for the measure of learning capability and describe both the evaluated ans\"atze and the numerical setup.
	We then show our main results in Sec.~\ref{Results} and discuss the implications in Sec.~\ref{Discussion}.
	
	\section{Prior work}
	\label{PriorWork}
	
	Investigating the model complexity of PQCs is a topic of active research.
	There are different angles and approaches one can take here.
	One important and active research field links model complexity to the generalization
	behavior~\cite{chen2021expressibility, Abbas2021, du2022efficient, bu2022,Caro2022outofdistribution,Caro2022, haug2023}, with only some works involving data reupload to obtain generalization bounds depending on the data reupload encoding scheme (see, for example, Ref.~\cite{caro2021}).
	Interestingly, Ref.~\cite{peters2022} demonstrates that, depending on the number of training samples, generalization is possible despite overfitting to training data.
	Furthermore, Ref.~\cite{gil2023} recently questioned the applicability of uniform generalization bounds by constructing a randomized supervised learning example, which reveals the memorization capability of PQCs.
	More work exists which investigates a model's complexity by assessing its capacity to memorize random data.
	Reference~\cite{wright2020capacity} compares the memorization capacity of quantum neural networks with classical neural networks based on the information-theoretic ideas introduced in Ref.~\cite{mackay2003} for classical perceptrons and generalized to neural networks in Ref.~\cite{friedland2018}.
	This notion got extended to storing quantum states in~\cite{lewenstein2021storage}.
	
	\subsection{Expressiveness and entanglement}
	
	In the following subsection, we present works in more detail, which investigate model complexity of PQCs in terms of their expressiveness and use this to reason about the performance of PQCs on QML tasks.
	An important approach in this area is the sampling of different parameters for a PQC and the calculation of the Kullback-Leibler divergence~\cite{Kullback1951} of the estimated fidelity distribution in comparison with the Haar-distributed ensemble~\cite{Sukin2019, Nakaji2021}.
	This gives a description of the expressiveness in the form that a small deviation from the Haar distribution means that all of the corresponding Hilbert space is equally reachable.
	However, this analysis does neither take into account the effect that input encoding has on the expressiveness of a circuit---it plays an essential role when classifying classical data as was proven in~\cite{Schuld2021}---nor the trainability, which is crucial for practical use cases.
	
	Other measures in this realm are the expressive power and effective dimension.
	Reference~\cite{du2020expressive} investigates the expressive power of PQCs as generative models for probability distributions in terms of the entanglement entropy.
	In Ref.~\cite{wu2021expressivity}, the authors look at the effect of repeating data-first encoding by introducing further qubits on learning physical observables and analyzing the expressiveness of PQCs via the Hilbert space dimension the quantum circuits act on.
	Reference~\cite{haug2021capacity} analyzes PQCs based on the effective dimension, which depends on the Fisher information matrix, and the parameter dimension, which the authors numerically estimate by evaluating the effective dimension for random initializations.
	
	In addition to the expressiveness, entanglement capability is another measure that is used to describe the performance of PQCs~\cite{Sukin2019, Rasmussen2020, Hubregtsen2021}.
	Results from Ref.~\cite{Hubregtsen2021} indicate that entanglement and expressiveness are only weakly correlated, and the authors of Ref.~\cite{Marrero2021} have found that entanglement can hinder trainability, e.g., by inducing barren plateaus.
	While this could lead one to the conclusion that a PQC should have only a limited amount of entanglement, a certain amount is considered as needed in order for quantum advantage to be possible.
	The works that investigate PQC performance use a variety of ways to construct entangling layers, ranging from linear and cyclic entangling with zero parameters~\cite{Schuld2021}, one parameter~\cite{Sukin2019}, or three parameters~\cite{Beer2021} to the layer-dependent strong entanglement described in Ref.~\cite{Schuld2020}.
	
	\subsection{Expressiveness regarding Fourier series}
	
	The current understanding of PQC evaluations representing Fourier series as described in Ref.~\cite{Schuld2021} is laid out in the following.
	The function $f(x)$ is defined as the expectation value of the result of several PQC evaluations:
	\begin{equation}
		\label{eq:fx_pqc_model}
		f_\Theta (x) = \bra{0} U^\dagger (x, \Theta) M U (x, \Theta) \ket{0}
	\end{equation}
	where $\Theta$ is a set of all trainable parameters, $x$ is a classical data input, $U$ is an arbitrary unitary, and $M$ is a measurement operator.
	For the usual \textit{layered} ansatz, the unitary $U$ takes the form
	\begin{equation}
		\label{Eq:CircuitStructure}
		U (x) = W_L S(x) W_{L-1} . . . W_1 S(x) W_0
	\end{equation}
	where the circuit is of depth $L$, $S$ is the data-upload unitary, and $W_l \defeq W_l(\boldsymbol{\theta}_l)$ are trainable unitaries with $\boldsymbol{\theta}_l \subset \Theta$ (see Fig.~\ref{fig:qvc_layer}).
	The result proven in Ref.~\cite{Schuld2021} is that this model can be described by a partial real-valued Fourier series:
	\begin{equation}
		f(x) = \sum_{\omega \in \Omega} c_\omega e^{i \omega x}
	\end{equation}
	where the Fourier coefficients $c_\omega = \cc{c}_{-\omega}$ are determined by the trainable unitaries and the measurement operator while the available frequency spectrum $\Omega$ is determined by the data (re-)upload operators.
	The frequencies $\omega$ are given by all possible sums of eigenvalues of the generating Hamiltonian $H$ of the input encoding $S(x) = e^{-i x H}$,
	such that $\Omega= [-\omega, ..., 0, ..., \omega]$.
	
	\begin{figure}[bt]
		\centering
		\begin{quantikz}
			\lstick{$\ket{0}$} & [0.2cm]
			\gate[style={fill=cvar!\intensityent}]{W} \qwbundle{n} &
			\gate[style={fill=cenc!\intensity}]{U_{\text{in}}}
			\gategroup[wires=1, steps=2, style={inner sep=4pt, dashed, fill=cbox!\intensityb}, background,
			label style={label position=below, yshift=-0.5cm}]
			{\text{repeat $L$ times}} &
			\gate[style={fill=cvar!\intensityent}]{W} &
			\meter{} & \qw
			\rstick{$\langle z\rangle$}
		\end{quantikz}
		\caption{General circuit layout for layered PQCs ans\"atze with $n$ qubits, $L$ layers and one $W$ in the zero layer [see Eq.~\eqref{Eq:CircuitStructure}]. The unitary gates $W$ depend on trainable parameters and $U_{\text{in}}$ encodes the input $x$.}
		\label{fig:qvc_layer}
	\end{figure}
	
	As an example, in the case of one layer and a unitary $U(x) = W_1 e^{-i \frac{x}{2} H} W_0$, the resulting model is a simple sinusoidal function if a Pauli-matrix with eigenvalues $\{\lambda_0,\lambda_1 \} =\{+1,-1\}$ is used as generator:
	\begin{equation}
		\label{eq:fx_singlequbit_singlelayer}
		f(x) = c_0 + 2 \abs{c_1} \cos (x- \arg(c_1))
	\end{equation}
	with $c_0=\sum_{i, i^\prime,j=0}^{1} \cc{W}_{0j} \cc{W}_{ji^\prime} M_{i^\prime i} W_{ij} W_{j0}$ and $c_1 = \sum_{i,i^\prime=0}^{1} \cc{W}_{01} \cc{W}_{1i^\prime} M_{i^\prime i} W_{i0} W_{00}$ where $\cc{W}_{ji}$ denotes the complex conjugate of $W_{ji}$ and the layer index of the unitaries $W$ is omitted because it follows from the order.
	
	Using the multi-index $\boldsymbol{j_l}$ over $n$ qubits for layer $l$, the sum over eigenvalues $\Lambda_{\boldsymbol{j_l}} = \sum_{q=1}^{n}\lambda_{ql}$, and the Einstein sum convention for indices, one finds that for $n$ qubits and $L$ layers the expectation value can be written as:
	
	\begin{equation}
		\begin{aligned}
			f(x) =
			& e^{i \frac{x}{2}(
				\Lambda_{\boldsymbol{k_1}} + \cdots + \Lambda_{\boldsymbol{k_L}}
				-
				\Lambda_{\boldsymbol{j_1}} - \cdots - \Lambda_{\boldsymbol{j_L}})}  \\
			& \times M^{\boldsymbol{i}^\prime}_{ \ \boldsymbol{i} }
			\cc{W}^{ \ \boldsymbol{k_L}}_{\boldsymbol{i}^\prime} \cdots
			\cc{W}^{ \ \boldsymbol{0}}_{\boldsymbol{k_1}}
			W^{\boldsymbol{i}}_{ \ \boldsymbol{j_L}} \cdots
			W^{\boldsymbol{j_1}}_{ \ \boldsymbol{0}}
		\end{aligned}
	\end{equation}
	
	The detailed calculation is documented in Example~\ref{ex:nqubits_Llayer} in Appendix~\ref{ap:calculations}.
	
	The observation that evaluations of PQCs yield Fourier series can be used to describe the expressiveness of the model via the available frequency range as well as the values of Fourier coefficients that can be reached.
	The latter one can be accessed via the inverse Fourier transform.
	For repeated single-qubit data encoding, i.e., the case where $S(x)$ is not changed throughout the circuit---as is often done in the literature for classical input encoding and also in this work---the size $K$ of the available independent nonzero frequencies is limited by
	\begin{equation}
		K = n L,
		\label{eq:max_available_frequency}
	\end{equation}
	where $L$ is the number of layers in the PQC and $n$ is the number of qubits; this means that the frequency spectrum is limited by the total number of repetitions of the data-encoding gate.
	Furthermore,  $2K + 1$ real-valued parameters are necessary to express $K$ nonzero complex coefficients and one real valued coefficient $c_0$.
	Hence, one needs at least $2K + 1$ parameters in the gates that form the trainable unitaries $W$ of Eq.~\eqref{Eq:CircuitStructure} to be able to have nonzero coefficients for all accessible frequencies from $0$ to $\omega_K$.
	
	Several works take advantage of the fact that PQCs represent Fourier series.
	For example, Ref.~\cite{fontana2022} proposes an error mitigation technique based on Fourier coefficients and frequencies.
	Additionally, Ref.~\cite{fontana2022efficient} shows that cost functions of training a QML model can be efficiently classically computed if the number of parameters scales polynomial in circuit depth.
	Reference~\cite{liao2022expressibility} uses the training of Fourier series to test different PQC architecture adjustments and define PQCs which fit certain Fourier series better more expressive.
		Reference~\cite{wen2024enhancing} define residual connections in quantum neural
		networks and show that these increase the number of generated frequencies
		and the flexibility in adjusting the Fourier coefficients.
		While most works focus on single-dimensional Fourier series, Ref.~\cite{casas2023multi} uses multidimensional truncated
		Fourier series to analyze the input encoding scheme and determine which
		PQCs provide enough degrees of freedom  for fitting these functions.
		
		In the following, we use the Fourier series character of PQC evaluations to test which ans\"atze provide a richer class of trainable functions.
		´
		\section{Methods and setup}
		\label{Methods}
		
		\subsection{Learning capability}\label{sec:lc}
		
		To improve the understanding of PQC performance, we extend the methodology first described in Ref.~\cite{Schuld2021} as laid out in the following.
		The learned model function $f_\Theta$ that results from training a PQC model [see Eq.~\eqref{eq:fx_pqc_model}] to approximate a given Fourier function $g$ is a measure of both the trainability and the expressiveness of the corresponding circuit, given a choice of training procedure.
		In our studies, we use the corresponding loss value given by the mean-squared error (MSE)
		
		\begin{equation}
			\label{eq:mse}
			\varepsilon_{g} = \frac{1}{|X|} \sum_{x \in X} \left[ f_\Theta (x) - g(x)\right]^2
		\end{equation}
		where $|X|$ is the size of the test set.
		The mean squared error is used here since it is a common metric in machine learning that is suitable to determine the difference between the results of the PQC model and the Fourier function.
		
		We use a set of functions $G_d$ which contains $|G_d|$ randomly sampled, normalized Fourier functions, i.e., trigonometric polynomials, of degree $d$ and define learning capability as the average over the final individual validation losses:
		
		\begin{equation}
			\label{eq:learning_capability}
			\mu_d = \frac{1}{|G_d|} \sum_{g\in G_d} \varepsilon_{g}.
		\end{equation}
		
		This measure enables the practical analysis of the performance of different PQC architectures that include data encoding; this is done by using the insight that these architectures can (only) learn Fourier functions to investigate expressiveness and trainability.
		We take the average over a set of Fourier functions $G_d$ to address statistical variance in the training procedure.
		This way, an ansatz with a lower numerical value for $\mu_d$ is on average better suited to learn Fourier series with a determinable confidence interval than an ansatz with a higher numerical value $\mu_d$.
		
		The MSE in Eq.~\eqref{eq:mse} is a common and natural choice as a
		loss function for regression problems. It connects the learning task considered in this work to the
		framework of maximum likelihood estimation, providing an objective function
		proportional to the negative log-likelihood, under the assumption of independent and
		normally distributed residuals~\cite{Murphy2012,Goodfellow2016,Bishop2023}. In principle,
		$\varepsilon_{g}$ could be replaced with a different loss, if a similar analysis would be
		considered on a different problem class, e.g., classification with generic data sets such as
		the ones suggested in Ref.~\cite{bowles2024} and a cross-entropy loss.
		
		The size of the set of random Fourier functions needs to be large enough to sufficiently
		characterize the performance of a given circuit architecture on the problem set. We chose
		$\abs{G_{d}}=100$ as this produces reasonably tight confidence intervals to assume an unbiased
		estimate and allow for qualitative statements. Smaller sizes of $G_{d}$ may be possible, but
		were not considered here; for considerably larger sizes, one would expect diminishing returns
		on the quality of the estimate of $\mu_{d}$.
		In Appendix.~\ref{ap:error_sampling} we provide a numerical comparison of the mean and
		confidence interval for different set sizes based on the first experiments presented in our Results section.
		The numerical comparison underlines that the choice of $\abs{G_d} = 100$ leads
		to a reasonably accurate estimate of the mean while still being technically feasible.
		
		An interesting question is if $G_{d}$ can be limited to a potentially small
		number of specific functions that are a good representative of a real-world problem or
		class of real-world problems. This question is not straightforward to answer
		and relates to the field of neural architecture search~\cite{Elsken2019} in
		classical machine learning, an active area of ongoing research. We discuss some aspects of
		this in Sec.~\ref{Discussion}.
		
		\subsection{Evaluated architectures}
		We evaluate the learning capability of different PQC ans\"atze and change circuit properties systematically to reveal the effects of each part.
		More precisely, we analyze different ways of how to construct the trainable $W$s in Fig.~\ref{fig:qvc_layer} and Eq.~\eqref{Eq:CircuitStructure}.
		
		In the literature, the most commonly used entanglement gates without trainable parameters are $\CNOT$ (e.g., circuits 2, 11, and 15 in Ref.~\cite{Sukin2019} as well as circuits in Ref.~\cite{Schuld2021}) and $\CZ$ (e.g., circuits 9, 10, and 12 in Ref.~\cite{Sukin2019} as well as circuits in Refs.~\cite{McClean2018, PerezSalinas2020, Holmes2022, Cerezo2021Cost, Skolik2022, Jerbi2021, Heimann2022}).
		In addition, controlled-$R_X$ ($\CRX$) (e.g., circuits 4, 6, 8, 14, 17, and 19 in Ref.~\cite{Sukin2019} as well as circuits in Ref.~\cite{Meyer2022}) or $R_Z$ (e.g., circuits 3, 5, 7, 13, 16, and 18 in Ref.~\cite{Sukin2019}) are often used as entanglement gates with one trainable parameter.
		We also test the $\CAN$ gate~\cite{Crooks2019} because it provides the interesting case of three trainable parameters and bridges the gap to dQNN ansatz structures. This gate is composed of three two-qubit rotations, $\CAN(\boldsymbol{\theta}) = R_{XX}(\theta_0)R_{YY}(\theta_1)R_{ZZ}(\theta_2)$.
		
		Besides the entangling gates, every unitary $W$ contains single-qubit gates.
		Different options range from single operations (e.g., $R_Y$ in Ref.~\cite{Cerezo2021Cost}) for more hardware-efficient ans\"atze, to two operations (e.g., $R_Y R_Z$ in Ref.~\cite{Skolik2022}) and to general operations with three single gates (e.g., $R_Y R_Z R_Y$ in Ref.~\cite{Schuld2021}).
		Thus, in this work, we use these typical single-qubit unitaries $U^1 \in \{R_Y, R_YR_Z, R_YR_ZR_Y\}$ together with two-qubit unitaries $U_{\text{ent}} \in \{\CZ,\CNOT, \CRX, \CAN\}$.
		
		In Ref.~\cite{Schuld2020}, strongly entangling circuits were introduced to increase the entanglement in shallow circuits.
		This technique includes splitting each $W$ in blocks, which we refer to as \textit{entanglement layers}, where rotation gates are applied at the beginning of each block followed by controlled two-qubit gates.
		In our analysis, we choose the entanglement range for strong entanglement in each block to be equal to the number of the current block, i.e., if we use three entanglement layers, the first block contains ordinary controlled operations with control range 1, then the second block with range 2 and the third block with range 3.
		We contrast this ansatz of creating strong entanglement with a simpler one, which has several entanglement layers but keeps the control range of the entanglement equal to 1 for each block.
		The difference is depicted in Fig.~\ref{fig:qvc_ent_layer}.
		
		\begin{figure}[bt]
			\centering
			\begin{adjustbox}{width=\linewidth}
				\begin{quantikz}
					&
					\gate[style={fill=cvar!\intensity}]{U^1}
					\gategroup[wires=4, steps=4, style={inner sep=4pt, dashed}, background,
					label style={label position=below, yshift=-0.5cm}]
					{el 1} &
					\ctrl{1} &
					\qw & \qw &
					\gate[style={fill=cvar!\intensity}]{U^1}
					\gategroup[wires=4, steps=4, style={inner sep=4pt, dashed}, background,
					label style={label position=below, yshift=-0.5cm}]
					{el 2}&
					\ctrl{1} &
					\qw & \qw & \qw & \\
					& \gate[style={fill=cvar!\intensity}]{U^1} &
					\targ{} &
					\ctrl{1} &
					\qw &
					\gate[style={fill=cvar!\intensity}]{U^1} &
					\targ{} &
					\ctrl{1} &
					\qw & \qw \\
					& \gate[style={fill=cvar!\intensity}]{U^1} &
					\qw & \targ{} &
					\ctrl{1} &
					\gate[style={fill=cvar!\intensity}]{U^1} &
					\qw & \targ{} &
					\ctrl{1} & \qw\\
					& \gate[style={fill=cvar!\intensity}]{U^1} &
					\qw & \qw & \targ{} &
					\gate[style={fill=cvar!\intensity}]{U^1} &
					\qw & \qw & \targ{} & \qw
				\end{quantikz}
				\begin{quantikz}
					&
					\gate[style={fill=cvar!\intensity}]{U^1}
					\gategroup[wires=4, steps=4, style={inner sep=4pt, dashed}, background,
					label style={label position=below, yshift=-0.5cm}]
					{el 1} &
					\ctrl{1} &
					\qw & \qw &
					\gate[style={fill=cvar!\intensity}]{U^1}
					\gategroup[wires=4, steps=4, style={inner sep=4pt, dashed}, background,
					label style={label position=below, yshift=-0.5cm}]
					{el 2}&
					\ctrl{2} &
					\qw & \targ{} & \qw & \\
					& \gate[style={fill=cvar!\intensity}]{U^1} &
					\targ{} &
					\ctrl{1} &
					\qw &
					\gate[style={fill=cvar!\intensity}]{U^1} &
					\qw &
					\ctrl{2} &
					\qw & \qw \\
					& \gate[style={fill=cvar!\intensity}]{U^1} &
					\qw & \targ{} &
					\ctrl{1} &
					\gate[style={fill=cvar!\intensity}]{U^1} &
					\targ{} & \qw &
					\ctrl{-2} & \qw\\
					& \gate[style={fill=cvar!\intensity}]{U^1} &
					\qw & \qw & \targ{} &
					\gate[style={fill=cvar!\intensity}]{U^1} &
					\qw & \targ{} & \qw & \qw
				\end{quantikz}
			\end{adjustbox}
			\caption{Simple (left) and strong (right) entanglement structure for entanglement layer one and two depicted for generic single-qubit unitaries $U^1$ and $\CNOT$ as an example for a two-qubit unitary.}
			\label{fig:qvc_ent_layer}
		\end{figure}
		
		Another property which is interesting to analyze is the amount of entanglement gates per entanglement layer.
		If this amount is equal to the number of qubits, then we arrange the entanglement gates descending and connect the last with the first qubit.
		This method is commonly used (e.g., circuits 10, 13, 14, 18, and 19 in Ref.~\cite{Sukin2019} as well as circuits in Refs.~\cite{Schuld2021, Skolik2022}).
		We refer to it as \textit{cyclic}.
		One can also use one entanglement gate less than the number of qubits, neglecting the last entanglement gate in the cyclic structure (e.g., circuits 2, 3, 4, 9, 16, and 17 in Ref.~\cite{Sukin2019} as well as circuits in Refs.~\cite{McClean2018, Holmes2022}).
		We call this setup \textit{linear} because it does not connect the last qubit with the first one in an entanglement structure with a range equal to one.
		This case is shown in Fig.~\ref{fig:qvc_ent_style}.
		
		\begin{figure}[bt]
			\centering
			\begin{adjustbox}{height=0.8cm} 
				\begin{quantikz}
					\qw & \ctrl{1} & \qw & \qw\\
					\qw & \targ{} & \ctrl{1} & \qw \\
					\qw & \qw & \targ{} & \qw
				\end{quantikz}\hspace{3cm}
				\begin{quantikz}
					\qw & \ctrl{1} & \qw & \targ{} & \qw\\
					\qw & \targ{} & \ctrl{1} & \qw & \qw \\
					\qw & \qw & \targ{} & \ctrl{-2} & \qw
				\end{quantikz}
			\end{adjustbox}
			\caption{Linear (left) and cyclic (right) entanglement style, with the $\CNOT$ as an example of a two-qubit unitary acting on the first and the last wire it overlaps with.}
			\label{fig:qvc_ent_style}
		\end{figure}
		
		As an example, if one chooses a simple, cyclic entanglement with one entanglement layer, $\CZ$ gates, and two rotation gates as a single-qubit unitary, then one arrives at the circuit ansatz chosen in Ref.~\cite{Skolik2022}.
		Furthermore, we consider circuits that are similar to the alternating layer ansatz~\cite{Cerezo2021Cost, Nakaji2021} and the hardware-efficient ansatz~\cite{Kandala2017, Nakaji2021} by using the entanglement gates to first connect qubits with even and then qubits with uneven index within an entanglement layer.
		
		\begin{figure}[bt]
			\centering
			
			\begin{subfigure}[b]{\linewidth}
				\centering
				\begin{adjustbox}{height=0.9cm} 
					\begin{quantikz}
						\lstick{$\ket{\boldsymbol{0}}$} & \gate[style={fill=cenc!\intensity}]{U_\text{in}} \qwbundle{i} & \gate[style={fill=cvar!\intensity}]{U^1} & \gate[style={fill=cvar!\intensityent}, wires=2]{U^2} & \qw \\
						\lstick{$\ket{\boldsymbol{0}}$} & \qwbundle{h_1} \qw & \qw & \qw & \gate[style={fill=cvar!\intensity}]{U^1} &
						\gate[style={fill=cvar!\intensityent}, wires=2]{U^2} & \qw \\
						\lstick{$\ket{0}$} & \qw & \qw & \qw & \qw & \qw & \gate[style={fill=cvar!\intensity}]{U^1} & \meter{} & \qw \rstick{$\langle z\rangle$}
					\end{quantikz}
				\end{adjustbox}
				\caption{dQNN circuit ansatz $[i,h_1,1]$ for complex-valued data as described by Beer et. al.~\cite{Beer2020,Beer2021} with $i$ input qubits, one hidden layer with $h_1$ qubits, and one output qubit.}\label{fig:dqnn_org}
			\end{subfigure}
			
			\vspace{0.5cm}
			
			\begin{subfigure}[b]{\linewidth}
				\centering
				\begin{adjustbox}{height=0.9cm} 
					\begin{quantikz}
						\lstick{$\ket{\boldsymbol{0}}$} & \gate[style={fill=cenc!\intensity}]{U_\text{in}} \qwbundle{i} & \gate[style={fill=cvar!\intensity}]{U^1} & \gate[style={fill=cvar!\intensityent}, wires=2]{U^2} & \qw \\
						\lstick{$\ket{\boldsymbol{0}}$} & \qwbundle{h_1} \qw & \qw & \qw &
						\gate[style={fill=cenc!\intensity}]{U_\text{in}} & \gate[style={fill=cvar!\intensity}]{U^1} &
						\gate[style={fill=cvar!\intensityent}, wires=2]{U^2} & \qw \\
						\lstick{$\ket{0}$} & \qw & \qw & \qw & \qw & \qw & \qw & \gate[style={fill=cvar!\intensity}]{U^1} & \meter{} & \qw \rstick{$\langle z\rangle$}
					\end{quantikz}
				\end{adjustbox}
				\caption{dQNN circuit ansatz with included data reupload structure.}\label{fig:dqnn_dr}
			\end{subfigure}
			
			\vspace{0.5cm}
			
			\begin{subfigure}[b]{\linewidth}
				\centering
				\begin{adjustbox}{height=0.8cm} 
					\begin{quantikz}
						\lstick{$\ket{\boldsymbol{0}}$} & \gate[style={fill=cvar!\intensity}]{U^1} \qwbundle{i} &
						\gate[style={fill=cenc!\intensity}]{U_\text{in}} & \gate[style={fill=cvar!\intensity}]{U^1} & \gate[style={fill=cvar!\intensityent}, wires=2]{U^2} & \qw \\
						\lstick{$\ket{\boldsymbol{0}}$} & \qwbundle{h_1} \qw & \qw & \qw & \qw &
						\gate[style={fill=cvar!\intensity}]{U^1} &
						\gate[style={fill=cenc!\intensity}]{U_\text{in}} & \gate[style={fill=cvar!\intensity}]{U^1} &
						\gate[style={fill=cvar!\intensityent}, wires=2]{U^2} & \qw \\
						\lstick{$\ket{0}$} & \qw & \qw & \qw &\qw & \qw & \qw & \qw & \qw & \gate[style={fill=cvar!\intensity}]{U^1} & \gate[style={fill=cvar!\intensity}]{U^1} & \meter{} & \qw \rstick{$\langle z\rangle$}
					\end{quantikz}
				\end{adjustbox}
				\caption{dQNN ansatz with data reupload and $U^1U_{\text{in}}U^1$ structure in each hidden layer.}\label{fig:dqnn_dr_zl}
			\end{subfigure}
			
			\vspace{0.5cm}
			
			\begin{subfigure}[b]{\linewidth}
				\centering
				\begin{adjustbox}{height=1.2cm} 
					\begin{quantikz}
						\lstick{$\ket{i_0}$} & \gate[style={fill=cvar!\intensityent}, wires=4]{U^2} & \qw \\
						\lstick{$\ket{i_1}$} & \qw & \qw \\
						\qw & \qw & \qw \rstick{$\ket{o_0}$} \\
						\qw & \qw & \qw \rstick{$\ket{o_1}$}
					\end{quantikz}
					=
					\begin{quantikz}[transparent]
						\lstick{$\ket{i_0}$} & \gate[style={fill=cvar!\intensityent}, label style={yshift=0.65cm}, wires=3]{\text{CAN}} & \gate[style={fill=cvar!\intensityent}, label style={yshift=1.3cm}, wires=4]{\text{CAN}} & \qw & \qw \qw & \qw\\
						\lstick{$\ket{i_1}$} & \linethrough \qw & \linethrough \qw & \gate[style={fill=cvar!\intensityent}, wires=2]{\text{CAN}} & \gate[style={fill=cvar!\intensityent}, label style={yshift=0.65cm}, wires=3]{\text{CAN}} \qw & \qw \\
						\qw & \qw & \linethrough \qw & \qw & \linethrough \qw & \qw \rstick{$\ket{o_0}$} \\
						\qw & \qw & \qw & \qw & \qw & \qw \rstick{$\ket{o_1}$}
					\end{quantikz}
				\end{adjustbox}
				\caption{For the dQNN ans\"atze, the unitary gate $U^2$ is composed of CAN gates.}\label{fig:dqnn_u2}
			\end{subfigure}
			
			\caption{Circuit representation and our modification of dQNN ans\"atze by Beer et. al.~\cite{Beer2020,Beer2021}.}\label{fig:dqnn}
		\end{figure}
		
		In addition to the layered architectures, we include the proposed circuit structure for dQNNs in Refs.~\cite{Beer2020,Beer2021} as an ansatz type.
		The original circuit of Ref.~\cite{Beer2021} is depicted in Fig.~\ref{fig:dqnn_org}.
		This architecture uses qubits as neurons in analogy to classical neural networks.
		Quantum perceptrons $U^l_j$ propagate the information
		through a unitary $U^2$ consisting of two-qubit $\CAN$ gates which connect
		each qubit of the previous layer $l-1$ with the $j$th neuron of layer $l$
		[see Fig.~\ref{fig:dqnn_u2}].
		In addition, parametrized single-qubit gates $U^1$ are applied to each qubit.
		Once the information is propagated, qubits of input and hidden layers are discarded, i.e., their values are not measured.
		The notation
		\begin{equation}
			[i, h_1, \cdots, h_H, \text{out}]
			\label{eq:notation_qubits_dqnn}
		\end{equation}
		represents the number of input qubits $i$, the number of qubits of each hidden layer $h_j$ for all $H$ hidden layers, and the number of output qubits $\text{out}$ which is equal to one in this work.

		We consider classical data $\{x, g(x)\}_{x\in X}$ which is encoded via
		$U_\text{in}(x)$ creating the
		input state $\rho^\text{in} = \ket{\phi^\text{in}_x} \bra{\phi^\text{in}_x}$
		with $\ket{\phi^\text{in}_x} = (U_\text{in}(x)\ket{0})^{\otimes i}$.
		The, possibly mixed, network's final output state can be expressed by
		\begin{align}\label{eq:dqnn_output_state}
			\rho^\text{out}_{x, \Theta} &= \tr_\text{in, hid}\left(
			\mathcal{U} \rho^{\text{in} \prime} \mathcal{U}^\dagger
			\right) \\ \nonumber
			&= \tr_\text{in, hid}(
			U^\text{out}_{m_\text{out}} \ldots U^1_1 \\ \nonumber
			& \qquad \qquad \quad
			\times \rho^\text{in} \otimes
			\ket{0 \ldots 0}_\text{hid, out}\bra{0 \ldots 0} \\ \nonumber
			& \qquad \qquad \quad
			\times U^{1\dagger}_1 \ldots U^{\text{out} \dagger}_{m_\text{out}}
			),
		\end{align}
		where each quantum perceptron $U^l_j = U^l_j(\boldsymbol{\theta}^l_j)$ depends
		on a set of trainable parameters $\boldsymbol{\theta}^l_j$.
		In Fig.~\ref{fig:dQNN_dr_zl}, this setup is called dQNN without zero layer (w/o zl).
		
		Inspired by the layered structure of PQCs explained above, we propose and compare
		enhancements to the dQNN ansatz structure. First, we introduce a data reupload
		scheme by allowing data encoding on every qubit but the output one.
		The network's final output state is obtained as in Eq.~\eqref{eq:dqnn_output_state}
		by using $\ket{\phi^\text{in}_x} = \ket{0}^{\otimes i}$ and
		the modified quantum perceptron $\tilde{U}^l_j(x, \boldsymbol{\theta}^l_j)$
		which includes parametrized single-qubit operation $U^1$ applied
		after input encoding unitary $U_\text{in}(x)$ on each qubit of the previous layer $l-1$.
		This ansatz is depicted in Fig.~\ref{fig:dqnn_dr}
		and referred to as dQNN reuploading without a zero layer (reup w/o zl).
		Second, to provide more degrees of freedom,
		we introduce single-qubit gates before each encoding gate
		which changes $\ket{\phi^\text{in}_x}$ to
		$\ket{\phi^\text{in}} =\ket{0}^{\otimes i}$ and the quantum perceptron to
		$\tilde{U}^l_j(x, \boldsymbol{\theta}^l_j)$ consisting of CAN gates linking all qubits
		of layer $l-1$ to the $j$th qubit of layer $l$ and
		$U^1(\boldsymbol{\theta}^{l-1}_{j,1}) U_\text{in}(x)
		U^1(\boldsymbol{\theta}^{l-1}_{j,0})$
		applied on all qubits except the output qubit.
		This setup is referred to as dQNN reuploading with zero layer (reup w zl).
		This way, the architecture gets closer to the layered ansatz type.
		The full modification is depicted in Fig.~\ref{fig:dqnn_dr_zl}.
		If input unitaries are not repeated in between those parametrized single-qubit gates
		and the only input encoding is on the input layer
		$\ket{\phi^\text{in}} = u(\boldsymbol{\theta}^{0}_{i,1})
		U_\text{in}(x) u(\boldsymbol{\theta}^{0}_{i,0}) \ket{0}^{\otimes i}$, then
		the setup is called dQNN with a zero layer (w zl) and the network's output state
		is obtained by Eq.~\eqref{eq:dqnn_output_state} with
		quantum perceptrons $\tilde{U}^l_j(\boldsymbol{\theta}^l_j)$.
		
		The final output of our circuits, Eq.~\eqref{eq:fx_pqc_model},
		is given by the expectation value of an observable
		$f_\Theta(x) = \langle M \rangle_{\psi(x,\Theta)} = \expval{M}{\psi(x,\Theta)}$.
		In density matrix representation for possibly mixed states $\rho^\text{out}_{x,\Theta}$
		as in Eq.~\eqref{eq:dqnn_output_state} that equation becomes
		$f_\Theta(x) = \langle M \rangle_{\rho^\text{out}_{x,\Theta}} = \tr(\rho^\text{out}_{x,\Theta}M)$
		which allows to rewrite the cost function in
		Eq.~\eqref{eq:mse} as
		\begin{align}
			\epsilon_g = \frac{1}{\abs{X}}\Sigma_{x\in X}
			\abs{\langle M \rangle_{\rho^\text{out}_{x,\Theta}} - g(x)}^2.
		\end{align}
	
		We outline that we do not use the parameter matrix multiplication update rule defined by Ref.~\cite{Beer2020} in our work to train the dQNN ans\"atze but rather stick to the gradient method which we use for layered ans\"atze
		and is explained in more detail in the next subsection.
		In the follow-up work~\cite{Beer2021}, the authors use a gradient descent method
		with a first-order approximation of the training cost's gradient.
		This way, they are able to approximate the gradient by evaluating the circuit on
		real quantum devices for slightly shifted parameters $(p-\epsilon)$ and $(p+\epsilon)$.
		
		Since the authors of Ref.~\cite{Schreiber2022} suggest that a possible quantum advantage for machine learning tasks can only lie in the training on real quantum hardware, the search for optimal update rules is a major task ahead and left for further studies.
		
		\subsection{Numerical setup}
		As is standard for the application of VQAs, we use a hybrid quantum-classical optimization scheme~\cite{Cerezo2021VQAs}.
		Numerical results in this work were obtained using TensorFlow Quantum~\cite{Broughton2020TensorflowQuantum}, which utilizes the quantum circuit simulator qsim~\cite{qsim2020}, and Cirq~\cite{Cirq}.
		Gradients are calculated with the Adjoint method~\cite{Plessix2006, Luo2020}, which is suited for analytic simulations and the standard differentiator in TensorFlow Quantum.
		All results shown in the main text of this paper are calculated using analytic expectation values; for a comparison to shot-based values, see Appendix~\ref{ap:shots}.
		Moreover, we use a simulated noise model based on calibration data of the IBM Quantum Falcon r5.11H Processsor to demonstrate that the learning capability can be determined in a noisy setup and that general trends and observations can survive, see Appendix~\ref{ap:noise}.

		For our trainings, we use the Adam optimizer~\cite{KingmaAdam2014} with a maximum of $360$ epochs.
		The learning rate is decreased from $0.5$ for the first $120$ epochs to $0.1$ for the second $120$ epochs and $0.05$ for the last $120$ epochs.
		In Appendix~\ref{hyper} we compare the results for different hyperparameter configurations
		to justify this choice.
		The train and validation data sets contain $50$, respectively, $100$, data points for functions of degree less than 10, respectively, greater or equal 10.
		The $x$-values of the test data are calculated by including the endpoint of the interval $[0,2\pi]$, whereas the endpoint is excluded in the validation data $[0,2\pi)$ to ensure different sets.
		During the training, the train data is split into batch sizes of $25$, resp. $50$, and validated after each epoch using the validation data.
		
		As a cutoff, we choose a loss value of $5\times10^{-5}$ where the training terminates.
		We consider fits reaching this value as quasiperfect because they correspond to an
		average absolute error of $0.7\%$ for each data point.  While this cutoff corresponds
		to overfitting, it reveals ans\"atze that are well suited to learn arbitrary partial
		Fourier series of a given degree. As the exact ground truth of the learning task for the
		learning capability is known, and we are interested in the model's ability to fit the
		Fourier functions exactly, such small absolute errors are sensible compared to a
		concrete real-world learning problem, where the generalization over potentially
		incomplete, corrupted, or noisy data has to be balanced and traded off with model
		accuracy on training data.
		
		We test on 100 randomly chosen normalized Fourier functions.
		This gives stable results for mean values and standard deviations of the loss values, which are independent of the random initial parameters and the set of Fourier functions for each training.
		In Appendix~\ref{ap:fourier} we provide details on the generation of the Fourier functions, their correlation, and the dependence of the learning capability on the set of Fourier functions.
		
		We provide a marker for an average loss value of $6.25\times10^{-4}$ depicted as a gray line in the plots for the learning capability because this level corresponds to an absolute average error for each function fit of less than $5\%$ of the maximum Fourier function value.
		Below that value, we regard average loss values as not statistically significant and view the corresponding architectures as equal regarding their learning capability since there is no clear way to distinguish between different loss values in this regime, and no practical reason to do so.
		
		Regarding the architectures, we always use $R_X$ gates to encode the real-valued input data and a $\sigma_z$-measurement on the last qubit for layered PQCs as well as for dQNNs.
		
		\section{Results}
		\label{Results}
		
		To get an intuition on how the learning capability of a PQC relates to the Fourier coefficient calculation and numerical methods introduced by Schuld et al.~\cite{Schuld2021}, we advise the reader to look into Appendix~\ref{Preliminaries}
		where introductory results for simple circuits for Fourier degree 2 are presented that motivate our method.
		Furthermore, we look into the minimal number of qubits and layers required to learn a specific Fourier function in that Appendix.
		
		In the following subsections, we demonstrate an example where
		two PQC architectures have the same sampled Fourier coefficients
		but differ in their ability to fit different Fourier functions, i.e.,
		they show differences in their learning capability
		in Sec.~\ref{subsec:schuld} and analyze
		various elements of PQC architectures in detail for layered ans\"atze in Sec.~\ref{subsec:wsw} and for different dQNN circuits in Sec.~\ref{subsec:dqnn}.
		
		\subsection{Comparison to sampling Fourier coefficients} \label{subsec:schuld}
		We extend the numerical results by Schuld et al.~\cite{Schuld2021}, and show that the learning capability provides insights that cannot be gained, for example, by sampling parameters of PQCs and calculating Fourier coefficients.
		
		\begin{figure*}[tb]
			\centering
			\includegraphics[width=0.76\textwidth]{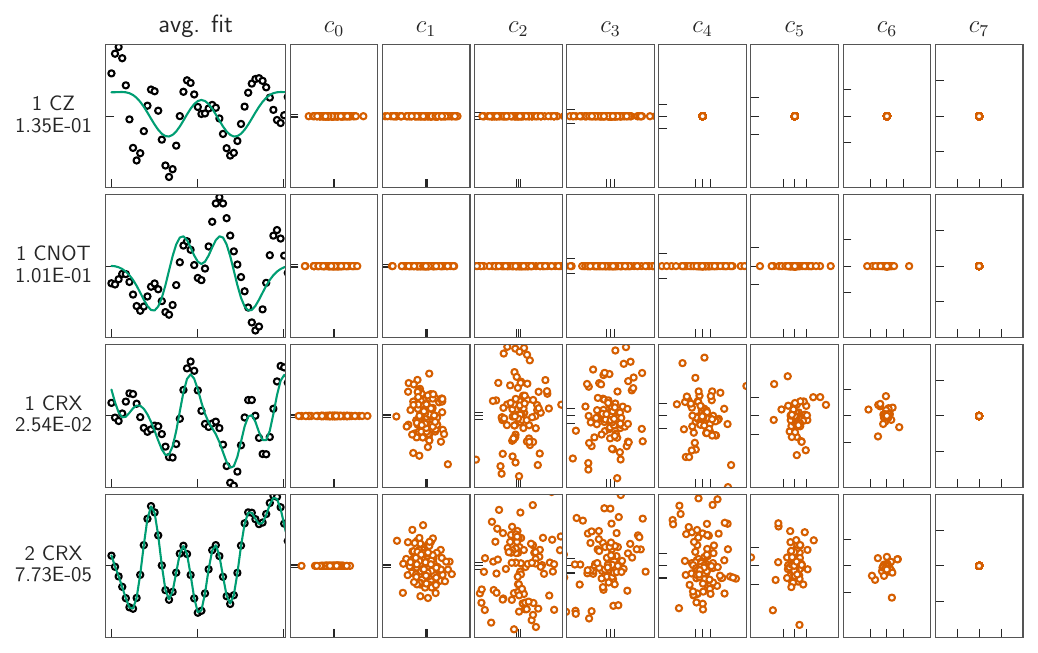}
			\caption{\textbf{Left column}: Validation results closest to the average loss value for four different layered ans\"atze with $n = 3$, $L = 2$, and $R_Y$ as the single-qubit rotation gate; the number of entanglement layers, the entanglement gate, and the average loss value are listed next to each graph.
			The results of the quantum circuits are shown as a green line and the Fourier functions as black circles.
			The $x$ axis in this column is from $0$ to $2\pi$, the $y$ axis from $-1$ to $1$.\newline
			\textbf{Right columns}: Fourier coefficients $c_0$ to $c_7$ resulting from the inverse Fourier transform of the results of evaluating the four ans\"atze with sampled parameters $\Theta$.
			The ticks on the $x$ and $y$ axes are always at $-0.01$, $0$, and $0.01$, showing how the Fourier coefficients become smaller with increasing degree due to the function values being limited to the range $[-1,1]$.
			The $x$ axis shows the real part of the coefficient, while the $y$ axis shows the imaginary part.\newline
			\textbf{Rows}: first: ansatz with CZ as the entangling gate and one entanglement layer; second: ansatz with $\CNOT$ as the entangling gate and one entanglement layer; third: ansatz with $\CRX$ as the entangling gate and one entanglement layer; fourth: ansatz with $\CRX$ as the entangling gate and two entanglement layers per $W$.}
			\label{fig:coeffs_funct}
		\end{figure*}
		
		As an example, we consider
		Fourier series of degree 6 and
		ans\"atze with three qubits, two layers, single-qubit operation $R_Y$, and a simple, linear entanglement structure with one entanglement layer.
		The complexity of the circuits makes analytic calculations already unfeasible and, hence, motivates the introduction of learning capability.
		The first row of Fig.~\ref{fig:coeffs_funct} corresponds to results for such an ansatz with CZ gates as entanglement gates, the second row uses $\CNOT$ gates, the third uses $\CRX$, and the fourth uses $\CRX$ gates but two entanglement layers instead of one.
		All circuits are fully depicted in Figs.~\ref{fig:wsw_circuit_q3l2_cz},~\ref{fig:wsw_circuit_q3l2_cx},~\ref{fig:wsw_circuit_q3l2_crx_el1}, and~\ref{fig:wsw_circuit_q3l2_crx_el2} in Appendix~\ref{circuits}.
		Because the learning capability is an average value, we select and plot in the first column of Fig.~\ref{fig:coeffs_funct} the learned model that has a loss value closest to the average.
		The details of the distribution of the errors are depicted in Fig.~\ref{fig:error_bardistribution} in Appendix~\mbox{\ref{ap:error}}.
		For the first two ans\"atze, determining the learning capability is not strictly necessary; however, it incorporates the fact that not all Fourier coefficients can be represented because the loss values for both ans\"atze are relatively high.
		The last two rows of Fig.~\ref{fig:coeffs_funct} show a case in which analyzing Fourier coefficients alone is not enough to describe how well circuits can represent Fourier functions since both ans\"atze give access to all possible coefficients.
		However, determining the learning capability reveals that the last ansatz, with two entanglement layers, fits Fourier functions of degree 6 on average much better than the other ans\"atze.
		
		
		How well an ansatz can fit Fourier functions depends on the ansatz, as well as the hyperparameters used for the training.
		However, we consider ans\"atze well suited in a practical sense, if they are robust to hyperparameter changes or if it is easier to find hyperparameters that lead to successful trainings.
		We provide a grid comparison of different hyperparameters for the third ansatz (three qubits, two layers, single-qubit operation $R_Y$, a simple, linear entanglement structure, one entanglement layer and $\CRX$ entanglement gates) in Appendix~\ref{hyper}.
		None of the $21$ different hyperparameters lead to a different learning capability.
		Hence, none of these parameter sets improves the learning capability of this ansatz.
		This fact strongly indicates that ans\"atze which enable the same Fourier coefficients but have different learning capabilities have some profound differences.
		

		\subsection{Layered ansätze} \label{subsec:wsw}
		
		\begin{figure}[t]
			\centering
			\includegraphics[width=\linewidth]{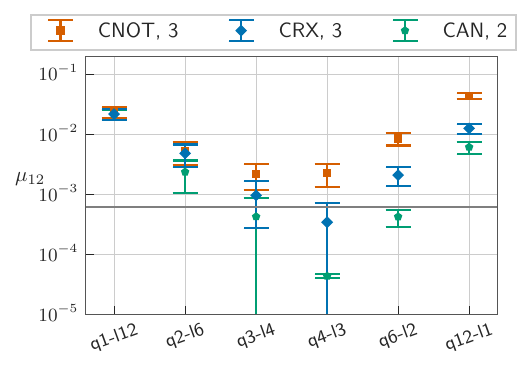}
			\caption{Learning capability $\mu_{12}$ for layered ans\"atze with single-qubits gates $R_YR_Z$ and simple, linear entanglement structure but different entanglement gates and layers. The different entanglement gates are $\{\CNOT, \CRX, \CAN\}$  with the number of entanglement layers $\{3,3,2\}$. The $x$ axis shows the architectures in terms of number of qubits and number of layers. An extended version of this figure is shown in Fig.~\ref{fig:pqc_d12_all} in Appendix~\ref{Appendix_plots}.}
			\label{fig:loss_degree12_entgates}
		\end{figure}
		
		Because we know that $nL = d$ is necessary to enable $d+1$ Fourier coefficients (including $c_0$), we start with analyzing the impact of different combinations of qubits and layers.
		For degree $12$ functions, for example, we can use the following ans\"atze: ($12$ qubits, $1$ layer), $(6,2)$, $(4,3)$, $(3,4)$, $(2,6)$, and $(1,12)$.
		If we choose layered architectures with simple, linear entanglement structure and parametrized single-qubit gates $R_YR_Z$, then we can vary the amount of entanglement layers as well as the entanglement gates.
		A summary of this comparison is depicted in Fig.~\ref{fig:loss_degree12_entgates}.
		The learning capability is plotted with the $95\%$ confidence interval for each architecture and the $5\%$ absolute error baseline is drawn in gray.
		First, we see that ans\"atze with $L\lesssim n$ have the best learning capability $\mu_{12}$, while ans\"atze with $L=d$ or $n=d$ show a poor learning capability $\mu_{12}$  with values of $0.1\geq \mu_{12} \geq 0.01$.
		
		In principle, barren plateaus could result in worse learning capabilities for ans\"atze with large $n$
		because, as can be seen in Fig.~3 in Ref.~\cite{McClean2018},
		varying the number of qubits between $1$ and $12$ reduces the variance
		of the gradients from $10^{-1}$ to $10^{-4}$, indicating a relatively strong decline.
		However, as can be seen in Fig.~4 in Ref.~\cite{McClean2018}, the variance of the gradients declines much less when varying the layers, especially for a small number of qubits (roughly $n<10$).
		From that figure, we would not expect to have vanishing gradients in our cases
		$(1, 12)$ and $(2,6)$. Hence,
		barren plateaus fail to explain the fact that ans\"atze with four qubits perform consistently better than ans\"atze with fewer qubits.
		Numerical simulations
		which are explained in more detail in Appendix~\ref{barren}
		confirm
		the conclusion that we draw from the figures in Ref.~\cite{McClean2018}
		regarding the barren plateaus for the PQCs that we consider.
		
		Second, we find that $\CAN$ gates with two entanglement layers perform slightly better than $\CRX$ with three entanglement layers and $\CNOT$ gates are not able to achieve reasonable results with two entanglement layers.
		However, due to the structure of the $\CAN$ gates, the corresponding architectures need more two-qubit gates despite having less entanglement layers (see Tab.~\ref{tab:parameter_count_WSW_d12} in Appendix~\ref{Appendix_plots}).
		
		The confidence interval for the configuration with three qubits,
		four layer and $\CAN$ gates with two entanglement layers
		and four qubits, three layer and $\CRX$ gates with three entanglement layers
		are visually large. In Appendix~\mbox{\ref{ap:error}}, especially Figs.~\ref{fig:error_bardistribution_q3_l4} and~\ref{fig:error_bardistribution_q4_l3},
		we plot the error distribution in
		comparison to results of different configurations.
		For those two examples, it becomes evident that the logarithmically scaled $y$-axis
		and cutting the plot's $y$-axis at $10^{-5}$ are the major reasons for the seemingly skewed and large confidence intervals.
		
		A more detailed comparison is given in Fig.~\ref{fig:pqc_d12_all} in Appendix~\ref{Appendix_plots}, which includes all results for linear and cyclic entanglement structure and two and three entanglement layers.
		These results show that the cyclic entanglement structure does not lead to a significant improvement in learning capability.
		This might be due to the fact that the amount of qubits is quite small but is nevertheless remarkable because it shows that for ans\"atze representing degrees $\leq 12$, cyclic entanglement structure is mostly not needed and thus, the gate count can be reduced.
		
		\begin{figure}[bt]
			\centering
			\includegraphics[width=\linewidth]{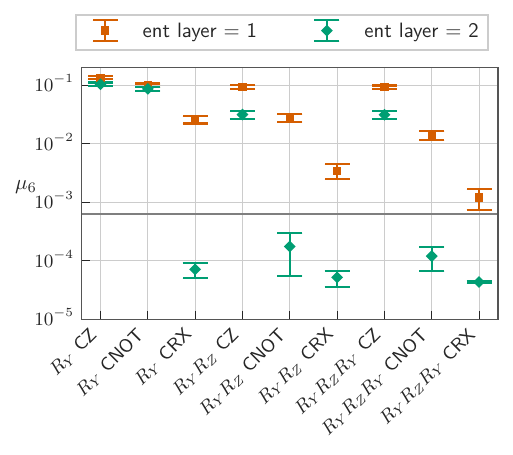}
			\caption{Comparison of $\mu_6$ for different layered ans\"atze with three qubits and two layers and simple, linear entanglement structure. The different $U^1 \in \{R_Y, R_YR_Z, R_YR_ZR_Y\}$ and entanglement gates $\{ \CZ, \CNOT, \CRX \}$ are split for different entanglement layers $\{1,2\}$. An extended version of this figure is shown in Fig.~\ref{fig:pqc_d6_cz_urot} in Appendix~\ref{Appendix_plots}.} \label{pqc_d5_cz_urot_q3l2}
		\end{figure}
		
		Next, we systematically analyze the impact of single-qubit unitaries and methods to construct the trainable $W$s.
		We choose degree $6$ Fourier functions because they are complex enough to reveal meaningful results about parameterized single-qubit operations and complement the previous results.
		We outline the results of three qubit, two layer ans\"atze in Fig.~\ref{pqc_d5_cz_urot_q3l2} as this combination follows the rule of $L \lesssim n$.
		Within each plot, all ans\"atze utilize a simple, linear entanglement structure.
		By varying $U^1 \in \{R_Y, R_YR_Z, R_YR_ZR_Y\}$,
		we find that using $R_YR_Z$ yields a satisfying learning capability in most cases, similar to $R_YR_ZR_Y$ and, in most cases, strictly better than using a single $R_Y$.
		Thus, the use of arbitrary rotations with three parameters is, in many cases, not needed, but using only one $R_Y$ is not enough to obtain good learning capabilities for layered ans\"atze.
		By varying the entanglement gates $\{\CZ, \CNOT, \CRX\}$ we show that, with our setup, CZ gates perform either as good as or worse than $\CNOT$ gates.
		Thus, we find that using $\CNOT$ gates as entanglement gates with no trainable parameter is preferable over CZ gates in a setup with $R_X$ encoding and $\sigma_z$-measurement on the last qubit.
		When varying the entanglement layer $\{1, 2\}$ it turns out that, for degree $6$ functions, using two entanglement layers is already enough for all considered entanglement gates.
		Remarkably, the ansatz with two linear entanglement layers, $\CRX$, and $R_Y$ has a high learning capability.
		However, this configuration had worse results for the learning capability on degree $12$ functions, showing that this configuration does not perform as well on higher order problems.
		CZ gates represent the only cases where a cyclic entanglement structure performs clearly better than a linear structure.
		However, even in these cases, the values for the learning capability only catch up to the values of the corresponding ansatz with $\CNOT$ gates (see Fig.~\ref{fig:pqc_d6_cz_urot} in Appendix~\ref{Appendix_plots}).
		A comparison of the results from three qubits, twolayers with six qubits, one layer
		is also given in Fig.~\ref{fig:pqc_d6_cz_urot}.
		It clearly indicates that three qubits, two layers performs much better than the other combination, similar to what we have seen for $d=12$.
		
		Since entanglement can play a larger role for larger qubit numbers, we increase the Fourier degree to 12 and test ans\"atze with four qubits and three layers.
		We choose ans\"atze with $U^1=R_YR_Z$ and entanglement gate $\CRX$ because they performed stable in previous experiments.
		We vary the entanglement
		structure $\{\text{simple}, \text{strong}, \text{strongc14}\}$, which denotes the styles shown in Fig.~\ref{fig:qvc_ent_layer} and the style used in circuit 14 from Ref.~\cite{Sukin2019}.
		The first three entries in Fig.~\ref{fig:pqc_d12_ent} show that these ans\"atze with two entanglement layers do not achieve a reasonable learning capability, no matter which structure we use.
		Instead, with three entanglement layers, all ans\"atze achieve an optimal learning capability independent of their entanglement structure and style.
		Furthermore, as the examples in Appendix~\ref{subsec:warm_up_1q1l} suggest---even beyond the case of one qubit and one layer---having a zero layer can enhance the learning capability drastically.
		All presented results demonstrate that increasing the number of entangling layers can change the learning capability quite drastically while not changing the Fourier coefficient picture.
		The results also indicate that the amount of entanglement layers is much more important than the structure that is implemented within each entanglement layer.
		
		\begin{figure}[bt]
			\centering
			\includegraphics[width=\linewidth]{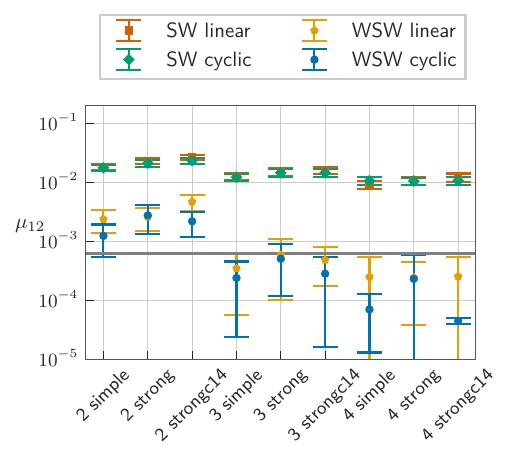}
			\caption{Learning capability $\mu_{12}$ for layered ans\"atze with four qubits and three layers,  $U^1=R_YR_Z$ and entanglement gate $\CRX$, either with zero layer ($WSW$) or without zero layer ($SW$) and linear or cyclic entanglement. The $y$ axis shows average loss values, while the $x$ axis shows the architectures with different entanglement layer numbers $\{2,3,4\}$ and entanglement structures $\{\text{simple}, \text{strong}, \text{strongc14}\}$. The latter corresponds to the entangling structure of circuit 14 in Ref.~\cite{Sukin2019}.} \label{fig:pqc_d12_ent}
		\end{figure}
		
		Using only half of the controlled operations in each entanglement layer, as done in the hardware-efficient or alternating layer ansatz reduces the learning capability; similar learning capabilities with these ans\"atze can be achieved by doubling the number of entanglement layers.
		This can be seen in Fig.~\ref{fig:d6_alt} where entanglement layers and gates vary as well as single-qubit operations.
		Furthermore, these results reveal the following three aspects of a circuit's depth:
		(1) architectures with four alternating entanglement layers with CRX and $R_Y$ or $R_YR_Z$ have a very low learning capability value, indicating that if one is interested in highly expressible PQCs the
		$R_Y$ has lower depth compared to $R_YR_Z$;
		(2) increasing the depth to six alternating entanglement layers does not improve the learning capability value,
		indicating that, in this setup, four alternating entanglement layers could be used for an architecture with lower depth; and
		(3) if the circuit must be short, i.e., contain two alternating entanglement layers, then utilizing the $\CRX$ gate
		would lead to PQCs with a lower learning capability value than $\CZ$ or $\CNOT$.
		
		The larger confidence interval for the configuration with $\CZ$ gates and two entanglement layers is partially due to the logarithmic scaling, but the confidence interval is indeed larger than, for example, the configuration with $\CRX$ gate and one entanglement layer because the result contains one outlier.
		The error distribution is plotted and compared to other results in Fig.~\ref{fig:error_bardistribution_q3_l2_alt} in Appendix~\ref{ap:error}.
		
		\begin{figure}[bt]
			\centering
			\includegraphics[width=\linewidth]{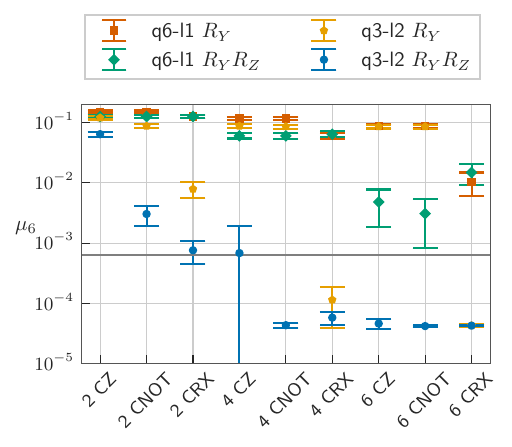}
			\caption{Learning capability $\mu_6$ of different alternating-layer (\textit{alt}) ans\"atze with six qubits, one layer and three qubits, two layers. The $x$-axis shows different entanglement layer numbers and entanglement gates per $W$.} \label{fig:d6_alt}
		\end{figure}
		
		To summarize our results for layered ans\"atze, we state that the best learning capabilities can already be achieved by using only a $\CNOT$ gate or the one-parameter $\CRX$ gate.
		The three-parameter $\CAN$ gate has a slightly better learning capability but also more parameters even though only two entanglement layers are necessary when using this gate for $d=12$ (see Figs.~\ref{fig:loss_degree12_entgates} and~\ref{fig:pqc_d12_all}).
		The entanglement style (linear or cyclic) as well as the entanglement structure (simple or strong) has almost no impact on $\mu_d$ for $d \leq 12$ (see Fig.~\ref{fig:pqc_d12_ent}).
		This shows that some necessary amount of entanglement has to be created by the ansatz.
		We establish that, for degree 6, two entanglement layers are enough for most architectures to gain sufficient learning capability (see Fig.~\ref{fig:coeffs_funct}) and three entanglement layers are enough for degree 12 functions (see Figs.~\ref{fig:pqc_d12_all} and~\ref{fig:pqc_d12_ent}).
		Increasing the entanglement layer count further does not improve the learning capability and, thus, can be considered inefficient for functions with equal or less degrees (see Fig.~\ref{fig:pqc_d12_ent}).
		
		\subsection{dQNN ansätze} \label{subsec:dqnn}
		
		\begin{figure}[bt]
			\centering
			\includegraphics[width=\linewidth]{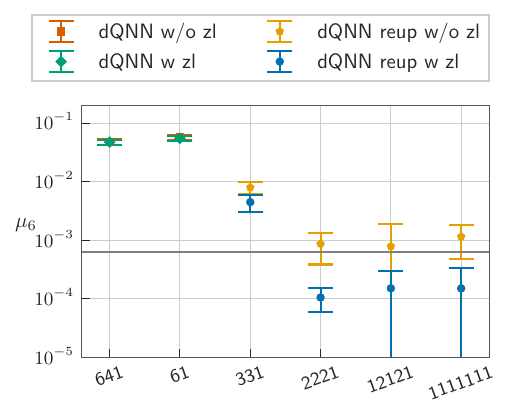}
			\caption{Comparison of learning capability $\mu_6$ of dQNN ans\"atze with $U^1 = R_YR_Z$ and without [see Fig.~\ref{fig:dqnn_dr}] or with [see Fig.~\ref{fig:dqnn_dr_zl}] zero layer (zl). The $x$ axis shows the configuration of the dQNN ans\"atze, where the first digit is the number of qubits in the input layer, the last digit is the number of qubits in the output layer, and the digits in between describe the number of qubits in the hidden layers. The first two architectures have no data reupload on their hidden qubits.} \label{fig:dQNN_dr_zl}
		\end{figure}
		
		The dQNN ansatz type is shown in Fig.~\ref{fig:dqnn} for the original circuit architecture from Refs.~\cite{Beer2020,Beer2021} and our enhancements, including the data reupload structure.
		Similar to the layered ans\"atze, one can see from Fig.~\ref{fig_d6_compare} in Appendix~\ref{Preliminaries_Multi} that a certain number of hidden qubits and layers is needed to achieve meaningful learning capability for a certain Fourier degree when using the dQNN ansatz type.
		In the case of $d=6$, a structure of at least $[2,2,2,1]$ was needed where the notation is according to Eq.~\eqref{eq:notation_qubits_dqnn}.
		This ansatz has two hidden layers and a maximum of four neighboring qubits, meaning qubits that are directly connected to each other.
		
		Evaluating $\mu_6$ for the dQNN ansatz type with $U^1 = R_YR_Z$, we find as a first result that the learning capability is close to $\mu_6 \approx 0.1$, and we are, thus, not able to learn Fourier series when we encode classical data only on the input qubits as depicted in Fig.~\ref{fig:dqnn_org}.
		This is shown in Fig.~\ref{fig:dQNN_dr_zl} where the first two ans\"atze $[6,1]$ and $[6,4,1]$ have enough qubits but not the right structure to achieve a meaningful learning capability.
		Note that the ansatz $[6,4,1]$ has a hidden layer of four qubits without data reupload.
		Defining a data reupload strategy for the hidden qubits increases the learning capability drastically.
		This is shown in Fig.~\ref{fig:dQNN_dr_zl} for both reupload ans\"atze from Fig.~\ref{fig:dqnn_dr} and Fig.~\ref{fig:dqnn_dr_zl}.
		Furthermore, we find that a medium-sized amount of  neighboring qubits (three or four in the case of a degree six function) is preferable over having many neighboring qubits (seven in the case of degree six) or only a few (two for degree six).
		If the number of neighboring qubits is higher (four or six instead of three for $\mu_6$), then the learning capability is better when more hidden than input qubits are used [see Fig.~\ref{fig:dQNN_dr_zl}].
		
		Next, we investigate the effect of different single-qubit gates on the learning capability of dQNN ans\"atze.
		First, we look at the effect of introducing a zero layer of parametrized rotations on each qubit, i.e., changing the architecture from the one in Fig.~\ref{fig:dqnn_dr} to the one in~\ref{fig:dqnn_dr_zl}.
		Starting with single-qubit parameterized rotations in a $U^1 U_{\text{in}}U^1$ structure instead of starting with the data-encoding in a $U_{\text{in}}U^1$ structure and including the parametrized rotations on each hidden qubit as well increases the learning capability consistently (see Fig.~\ref{fig:dQNN_dr_zl}).
		However, compared to the layered ans\"atze, the effect of this change is not as significant, and basic learning capability can already be achieved without a zero layer.
		
		\begin{figure}[bt]
			\centering
			\includegraphics[width=\linewidth]{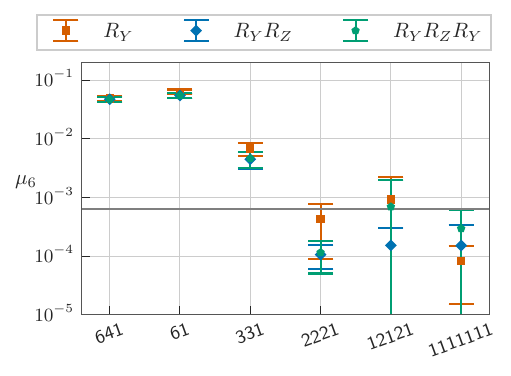}
			\caption{Comparison of learning capability of dQNN ans\"atze with zero layers according to Fig.~\ref{fig:dqnn_dr_zl} with different single-qubit unitary gates. The $y$ axis shows the learning capability as measured via the average mean squared error. The $x$ axis shows the configuration of the dQNN ans\"atze, where the first digit is the number of qubits in the input layer, the last digit is the number of qubits in the output layer, and the digits in between describe the number of qubits in the hidden layers. Ans\"atze with different single-qubit unitaries are represented in different colors.} \label{fig:dQNN_urot}
		\end{figure}
		
		As a second step, we compare the learning capabilities for dQNN ans\"atze with different types of parametrized rotations in Fig.~\ref{fig:dQNN_urot}.
		We find that using single-qubit gates $R_YR_Z$ yields a satisfying learning capability, which is similar to an ansatz using  $R_YR_ZR_Y$.
		However, a single $R_Y$ gate already yields satisfying results in many cases.
		Thus, the use of rotations with three or even two parameters is not needed in our test cases for dQNN ans\"atze, which could be due to the rich entanglement structure of the dQNN ans\"atze resulting from the $\CAN$ gates.
		
		\section{Discussion and Outlook}
		\label{Discussion}
		In this paper, we have defined the new measure of learning capability, which quantifies how well a PQC ansatz is able to learn Fourier functions of a specific degree on average.
		This quantity turns out to be necessary since existing measures for PQC performance, like the expressiveness in terms of the Haar measure~\cite{Sukin2019}, in terms of Fourier series coefficients~\cite{Schuld2021}, or in terms of dynamical Lie algebras (see Appendix.~\ref{ap:lie_algebra}), miss differences between certain ansatz types or lack a quantification of the trainability.
		The learning capability experiments reveal that PQCs exist, which provide
		all Fourier coefficients for the respective degree,
		but differ in their ability to fit different Fourier functions, i.e.,
		they show differences in their learning capability.

		The results in Sec.~\ref{subsec:schuld} show that even though the inverse Fourier transform of a circuit sampling provides all the Fourier coefficients for the respective degree, the learning capability may still be small, which extends the numerical results provided in Ref.~\cite{Schuld2021}.
		One reason for different learning capabilities in these cases could be due to interdependencies between different Fourier coefficients that hinder access to arbitrary Fourier functions.
		
		We have used the learning capability of PQCs to compare several layered and dQNN type ans\"atze.
		Thereby, we hope to build a bridge between theoretical analysis and concrete machine learning tasks that takes into account the input encoding, expressiveness, entanglement, and trainability.
		This is especially important considering the gate count reduction that is needed in the NISQ era and the different native gate sets that are available on different quantum computing hardware.
		When proposing new PQC architectures, the learning capability can be used to determine their performance compared to existing architectures.
		
		In general, we find that both layered and dQNN ans\"atze have a similar performance for a similar parameter or gate count (see Table~\ref{tab:parameter_count_WSW_dQNN}).
		For layered ans\"atze, the results suggest that an efficient architecture with high learning capability $\mu_d$ for Fourier functions with degree $d\in \{6,12\}$ can be obtained by using a similar number of qubits and layers or slightly more qubits than layers which can be stated by $n \gtrsim \sqrt{d}$ and $L \lesssim \sqrt{d}$.
		Furthermore, two rotation gates $R_YR_Z$ and the $\CNOT$ or $\CRX$ gate with a simple, linear entangling structure yield good learning capability.
		Within this setup, using more entanglement layers within one trainable unitary $W$ increases the learning capability until a saturation limit is reached which depends on the degree that is possible due to the data reupload structure.
		Our results show that two (respectively, three) entanglement layers in each $W$ yield good learning capabilities for degree six (respectively, 12) functions.
		
		For dQNN ans\"atze, a data reupload technique, inspired by the reupload technique for layered PQCs, turned out to be necessary.
		While these ans\"atze allow for many different amounts of input and hidden qubits, it becomes apparent that having more hidden than input qubits leads to better learning capabilities $\mu_6$ and that the hidden qubits should be structured in several hidden layers.
		In contrast to layered ans\"atze, the use of $R_Y$ as a single-qubit gate is sufficient to achieve good learning capability $\mu_6$ for dQNN ans\"atze.
		
		Overall, the learning capability of a chosen circuit ansatz is inherently limited by the structure and gate usage of the ansatz.
		Thus, a careful selection of the ansatz type and the number of qubits and layers is needed to build efficient circuits.
		When designing circuits for practical use cases, those numbers determine the maximum Fourier degree that can be learned and one should select the PQC structure and gates such that the learning capability is suitable.
		For example, the learning capability can be used to compare different ans\"atze
		with a similar value and identify the architectures that need
		fewer single-qubit operations, fewer entanglement layers, or a different qubit-to-layer ratio. This way, the depth can be reduced, leading to reduced barren plateau effects without lowering the capability to learn Fourier series.
		
		Because discrete Fourier series can approximate arbitrary functions up to a certain precision, we conclude that the performance of PQCs in learning Fourier functions
		can give some guidelines on the performance in general learning tasks.
		For example, we expect the architectures that learn Fourier series
		with low degrees inaccurately, i.e., with high learning capability values,
		to be unsuited for supervised learning tasks with a complex dependency between
		input and output data.
		We also expect architectures that learn a variety of Fourier series
		with a given degree accurately to be more likely
		to capture complex dependencies between input and output data.
		However, we emphasize that this assumption about the correlation between performance on
		Fourier series and general QML tasks, of course, may not hold universally.
			
		Under this assumption,
		our results can be used as guidance regarding this selection, meaning that a well-informed initial guess for the architecture of a PQC can be chosen by involving the learning capability.
		We concretize this by considering two different problems.
			
		First, we assume that very little is known about the solution of a supervised learning problem.
		Similar to the arguments presented in Ref.~\cite{friedland2018},
		we advise starting the training of a QML model with an ansatz with a good learning capability, i.e., choosing layered PQCs with $L\lesssim n$, $R_YR_Z$, and $\CRX$ gates,
		and to seek for minimizing the training error.
		By choosing the number of qubits and layers at the beginning of the procedure,
		the number of input gates is determined and hence,
		the maximal possible Fourier degree is known.
		The entanglement layers can be selected accordingly in the sense that two (respectively, three)
		are sufficient for six (respectively, 12) input encoding gates.
		More entanglement layers should be considered if more Fourier degrees are targeted.
		The number of qubits and layers (and entanglement layers) can successively be increased
		if the ansatz is not able to achieve low errors on the training data.
		Once an architecture has a low error on the training data, the evaluation on the test data determines if the architecture is suitable for the given task.
		In case of a large error on the test data, overfitting might be occurring and
		standard machine learning techniques like regularization could be applied as well as
		choosing PQC ans\"atze that reduce the depth, i.e., contain fewer entanglement layers,
		or those that reduce the gate complexity for compiling them on quantum hardware,
		i.e., use $\CNOT$ gates instead of $\CRX$, to reduce the overfitting.
		More practically speaking, finding a model that is suited for the training and
		test data is done simultaneously and by (automatized) hyperparameter search.
		Hyperparameter search grows exponentially in the number of considered parameters.
		Therefore, being able to reduce these parameters---e.g.,
		by excluding the parametrized single-qubit operation by restricting it to $R_YR_Z$
		and reducing the considered combination of $n$ and $L$ by excluding the extreme cases---provides benefits to hyperparameter search.
		
		Second, we present nonlearning problems in quantum computational fluid dynamics (qCFD);
		a growing field of interest (see, e.g.,
		Refs.~\cite{gaitan2020, Kiffner2023})
		that investigates the simulation of nonlinear classical systems on quantum computers.
		In CFD, Fourier series are known to describe initial conditions that need to be 
		accurately represented by quantum states or circuits.
		Moreover, spectral methods in fluid dynamics allow to express flow solutions as coefficients
		for ansatz functions such as Fourier series~\cite{canuto2007}.
		Our analysis and the learning capability provide concrete results on the question which
		architectures are more likely to learn Fourier series accurately.
		Hence, studying PQCs for state preparation of CFD-relevant initial states
		is a possible area of application.
		In addition, 
		in Ref.~\cite{Lubasch2020}, the authors propose a method that includes quantum variational computing
		to solve nonlinear problems. The
		authors propose matrix product states to represent energy potentials
		in quantum circuits
		and use them for the case where the energy potential is described by a Fourier series.
		Therefore, investigating PQCs for representing these energy potentials is a similar
		possible area of application
		where our work provides guidelines on the choice of architecture.
		
		In this work, we do not change the input encoding gate and use $R_X$ for one-dimensional classical data in all experiments.
		However, evaluating the learning capability of PQC ans\"atze on multi-variate functions is an important next step to move towards common machine learning use cases.
		This especially means designing and analyzing those PQCs
		which provide enough degrees of freedom required for
		fitting multi-variate Fourier series~\cite{casas2023multi}.
		Analyzing the effect of different input encodings per layer could be interesting for further studies.
		Additionally, one could also perform an in-depth investigation of the similarities between layered and dQNN ans\"atze.
		
		The analysis of different gradient methods, like the quantum natural gradient~\cite{Stokes2020} and the behavior of both training and evaluation on real quantum hardware are important topics that need further research.
		Finally, the learning capability could be used to analyze techniques like entanglement dropout~\cite{Kobayashi2022} and determine circuit ans\"atze which are suited to reduce overfitting on real-world, noisy input data.
		For finding a possible quantum advantage in learning~\cite{Schreiber2022} or generalization~\cite{Caro2022}, both research directions could be helpful.
		
		\section{Code}
		\label{Code}
		Code accompanying this paper is given in Ref.~\cite{HeimannData2023}.
		In this resource, we provide code to define an ansatz and determine its learning capability.
		The code can be used to perform the calculations presented in this paper and reproduce each reported value individually.
		We also provide code to create new sets of Fourier series and calculate their cross-correlations.
		
		\section{Acknowledgments}
		The authors thank Lukas Groß, Felix Wiebe, and Patrick Draheim for helpful	discussions	and the anonymous reviewers for their constructive remarks and recommendations to improve this manuscript.
		
		We acknowledge support from the Bundesministerium für Bildung und Forschung (BMBF) through the project Q$^3$-UP! under grant numbers 40 301 121 and 13 N 15 779  administered by the VDI/VDE Innovation + Technik GmbH (VDI) and by the Bundesministerium für Wirtschaft und Klimaschutz (BMWK) through the project QuDA-KI under the grant numbers 50RA2206A and 50RA2206B administered by the Deutsches Zentrum für Luft- und Raumfahrt e.V (DLR). We acknowledge the use of publicly available IBM Quantum services for this work. The views expressed are those of the authors.
		
		\newpage
		\clearpage
		\onecolumngrid
		\appendix
		
		\section{Dynamical Lie Algebras of PQCs}
		\label{ap:lie_algebra}
		In what follows, we will show that the parameterized quantum circuits used in our study have isomorphic dynamical Lie algebras (DLA). We will focus on parameterized quantum circuits with data reupload. A single data reupload layer, denoted by $U_l(\mathbf{\theta_l}, \mathbf{x_l})$, is then composed of an encoding layer, denoted by $U_l^e(\mathbf{x_l})$ for some input $\mathbf{x_l}$, a parameterized layer, denoted $U_l^p(\mathbf{\theta_l})$ for some parameters $\mathbf{\theta_l}$, and an entangling layer, denoted $U_l^{ent}$. In this case, we have
		\begin{align}
			U_l(\mathbf{\theta_l}, \mathbf{x_l}) &=
			U_l^{\text{ent}} U_l^p(\mathbf{\theta_l})
			U_l^e(\mathbf{x_l})\\
			U_l^e(\mathbf{x_l}) & = \prod_{k = 1}^{K} e^{iH_k^ex_{lk}} \\
			U_l^p(\mathbf{\theta_l}) & = \prod_{m = 1}^{M} e^{iH_m^p\theta_{lm}} \\
			U_l^{ent} & = \prod_{r = 1}^{R} e^{iH_{r}^{ent}} \\
			U(\mathbf{\theta}, \mathbf{x}) &= \prod_{l=1}^L U_l(\mathbf{\theta_l}, \mathbf{x_l})
			U_0^{ent}
			U_0^p(\mathbf{\theta_0})
		\end{align}
		where $U(\mathbf{\theta}, \mathbf{x})$ represents the unitary matrix of the whole circuit, and $L$ is the total number of layers.
		
		Now that we have our notation set for the circuit itself, we will remind you of the necessary definitions for computing the dynamical Lie algebra.
		
		\begin{defn} Given a parameterized quantum circuit $U(\mathbf{\theta}, \mathbf{x})$, let $\mathcal{G}_e = \{H^e_k\}_{k = 1}^K$ be the set of generators of an encoding layer, $\mathcal{G}_p = \{H^p_m\}_{m = 1}^M$ be the set of generators of a parameterized layer, and $\mathcal{G}_{ent} = \{H^{ent}_n\}_{r = 1}^R$ be the set of generators for an entangling layer.  Then, we define $\mathcal{G} = \mathcal{G}_l = \mathcal{G}_e \bigcup \mathcal{G}_p \bigcup \mathcal{G}_{ent}$ to be the set of traceless Hermitian matrices that generate the unitaries of a single layer.
		\end{defn}
		
		\begin{defn}
			Given the set of generators $\mathcal{G}$ of some parameterized quantum circuit, the dynamical Lie algebra $\mathfrak{g}$ is the Lie algebra spanned by the repeated nested commutators of the elements in $\mathcal{G}$:
			\begin{equation}
				\mathfrak{g} = \operatorname{span}\langle iH_0, ... iH_t \rangle_{\text{Lie}}
			\end{equation}
		\end{defn}
		
		\begin{lem}\label{lemmaD1} Let $\mathcal{G}_p = \{O_i\}_{i = 1}^n$, $\mathcal{G}_e = \{P_i\}_{i = 1}^n$, where $n$ is the number of qubits, $P_i \neq O_i \in \{X_i,Y_i,Z_i\}$ for all $i \in \{1,2,...,n\}$, and $\{X_i,Y_i,Z_i\}$ are the Pauli matrices on qubit $i$. The Lie algebra generated by the $\mathcal{G}_p \bigcup \mathcal{G}_e$ is
			\begin{equation}
				\mathfrak{g}' = \underbrace{\mathfrak{su}(2) \oplus \mathfrak{su}(2) \oplus ... \oplus\mathfrak{su}(2)}_{\text{n times}}.
			\end{equation}
		\end{lem}
		
		\begin{proof} The proof follows immediately from the defining relations of $\mathfrak{su}(2)$ and the definition of a dynamical Lie algebra.
		\end{proof}
		In the case of a linear entangling layer, the generators of the entangling gates take the following form:
		
		\begin{equation}
			\mathcal{G}_{ent} = \{U_{i+1}-Z_iU_{i+1}\}_{i = 1}^{n-1}
		\end{equation}
		This follows from the fact that we can write any controlled gate $CU$ as
		\begin{equation}
			CU = e^{i(I - Z_1)U_2},
		\end{equation}
		for some Hermitian matrix $U_2$.
		\begin{lem}\label{lemmaD2}
			Let $\mathcal{G}_p$, $\mathcal{G}_e$  be as in the previous Lemma~\ref{lemmaD1}, and let $\mathcal{G}_{ent} = \{U_{i+1}-Z_iU_{i+1}\}_{i = 1}^{n-1}$, then the Lie algebra generated by $\mathcal{G}_p \bigcup \mathcal{G}_e\bigcup G_{ent}$ contains all nearest neighbor two-body interactions.
		\end{lem}
		\begin{proof}
			From the previous Lemma~\ref{lemmaD1}, we know that we have all single-qubit Pauli matrices. Then, we can generate all two-body interactions by computing $[U_{i+1}-Z_iU_{i+1}, X_i]$, followed by all possible commutators of the result with all single-qubit Pauli matrices.
		\end{proof}
		
		\begin{lem}\label{lemmaD3}
			The following holds:
			\begin{align}
				&\left[M_iN_{i+1}, O_{i+1}P_{i+2} \right] = \alpha M_iQ_{i+1}P_{i+2} \\
				&\left[ \left[M_iN_{i+1}, O_{i+1}P_{i+2} \right], Q_{i+1}R_{i+2}\right] = \beta M_iS_{i+2}
			\end{align}
			for all $M,N,O,P,Q,R,S \in \{X, Y, Z\}$, $\alpha, \beta \in \mathbb{C}$, and $N\neq O$.
		\end{lem}
		\begin{proof}
			The first equality follows immediately from the defining relations of $\mathfrak{su}(2)$. The second equality follows from the first equality as well the defining relations of $\mathfrak{su}(2)$.
		\end{proof}
		
		\begin{prop}\label{prop:lie_algebra}
			Let $\mathcal{G}_p$, $\mathcal{G}_e$  be as in the previous Lemma~\ref{lemmaD1}, and let $\mathcal{G}_{ent} = \{U_{i+1}-Z_iU_{i+1}\}_{i = 1}^{n-1}$, then the Lie algebra generated by $\mathcal{G}_p \bigcup \mathcal{G}_e\bigcup G_{ent}$ contains all possible Pauli strings of length at most equal to the number of qubits n and the dynamical Lie algebra is $\mathfrak{su}(2^n)$.
		\end{prop}
		\begin{proof}
			From the previous Lemmas, we can see that by repeating this argument for next-to-nearest neighbor interactions then next-to-next-to nearest neighbor  interactions until we have all 2-body interactions as well as nearest neighbor 3-body interactions. Clearly, a similar argument can be made to generate all 3-body interactions and then $n$-body interactions.
		\end{proof}
		
		\begin{cor}\label{cor:lie_algebra}
			The dynamical Lie algebra of any circuit presented in this paper is always given by $\mathfrak{su}(2^n)$, where $n$ is the number of qubits.
		\end{cor}
		\begin{proof}
			The proof follows by applying Proposition~\ref{prop:lie_algebra} to any circuit.
		\end{proof}
		
		Corollary~\ref{cor:lie_algebra} shows that all our considered circuits tend to exhibit barren plateaus for a large number of layers.

		\section{Comparison with shot-based expectation values} \label{ap:shots}
		\begin{figure}[t]
			\centering
			\includegraphics[width=0.5\textwidth]{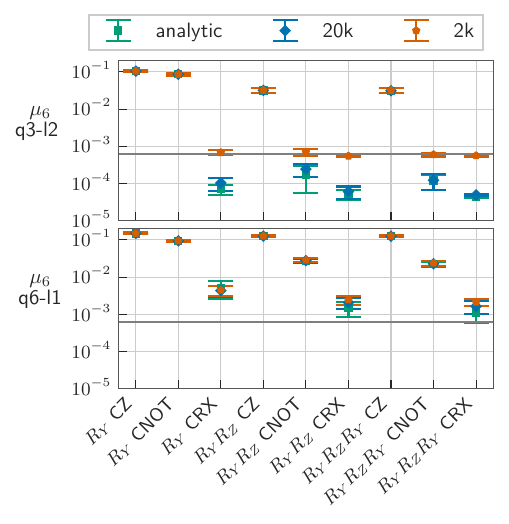}
			\caption{Comparison of $\mu_6$ for different layered ans\"atze with two simple, linear entanglement layers. The top plot shows results for three qubits and two layers, while the bottom plot shows results for 6 qubits and one layer. The green data points show analytic simulation results as in Fig.~\ref{fig:pqc_d6_cz_urot}. The blue data points depict results of shot-based expectation values using $20000$ shots and the parameter shift rule, while the orange data points show this for $2000$ shots.} \label{fig:wsw_degree6_mcshots}
		\end{figure}
		
		We consider the case of shot-based determination of expectation values and the parameter shift value to determine gradients.
		The analytic expectation values are computed using the noiseless backend of TensorFlow Quantum with no repetitions, while the shot-based expectation values are computed with a given number of samplings from the state.
		Having the training and execution shot-based brings our setup closer to real world execution while remaining hardware independent.
		
		Figure~\ref{fig:wsw_degree6_mcshots} shows, for three qubits and two layers, that the analytic values can be reached with $20000$ shots per expectation value.
		If we decrease the number of shots to a value lower than $20000 = \frac{1}{5\cdot10^{-5}}$, then the maximum reachable precision is lower then the one chosen for our analytic experiments.
		This is evident in the results for $2000$ shots where we get very similar results for architectures with high loss values, but are limited to a lower bound for the precision at $\sim 6.25 \times 10^{-4}$ (see values around the gray line in the upper row of Fig. \ref{fig:wsw_degree6_mcshots}).
		Comparing to possible experiments, we see that the shot-based expectation value converges to the analytical one, and that using a small number of shots would limit the maximum precision and would make circuit architectures with near-optimal learning capabilities indistinguishable.
		For six qubits and one layer, Fig.~\ref{fig:wsw_degree6_mcshots} shows that the values are very similar which is possible because the analytic values lay above the precision bound for $2000$ shots $\sim 6.25 \times 10^{-4}$.
		
		\section{Evaluation under noise}
		\label{ap:noise}
		
		To explore the usefulness of the proposed learning capability metric in an actual NISQ setting, we
		conduct exploratory experiments under simulated noise conditions that approximate the behavior of a
		current superconducting quantum processor. We opted for simulated noise over training on actual
		quantum hardware for two reasons. First, access to quantum hardware is limited and since the
		training and evaluation steps outlined below require millions of circuit evaluations, such a setting
		would be impractical as well as time and cost prohibitive. Second, execution on quantum hardware
		would necessitate transpilation to the native gate set of the chosen quantum processor. Depending
		how well certain gates can be expressed in that native gate set, results may be biased by an
		arbitrary choice of a quantum processor. Using a simulated device enables us to evaluate the exact
		same circuits and gates as in the noise-free setting.
		
		\subsection{Noise model implementation}
		\label{ap:noise_model}
		
		We implemented a noise model in Cirq~\cite{Cirq} based on a calibration data snapshot of the
		\textit{ibmq\_perth} system, which is one of the IBM Quantum Falcon~r5.11H Processors. This system
		has seven qubits, of which we copied five to extend the simulated device to 12 qubits.
		Figure~\ref{fig:noisy_grid} shows the topology of the simulated device. The qubits with indices
		$0\ldots 6$ correspond to the original hardware, whereas qubits $7\ldots 11$ are mirrored from the
		five rightmost qubits of the original device, including their coupling properties.The relevant
		calibration metrics obtained from the quantum processor, that we use in our noise model, are
		summarized in Table~\ref{tab:noisy_cal}.
		
		\begin{figure}[t]
			\centering
			\includegraphics[scale=0.6]{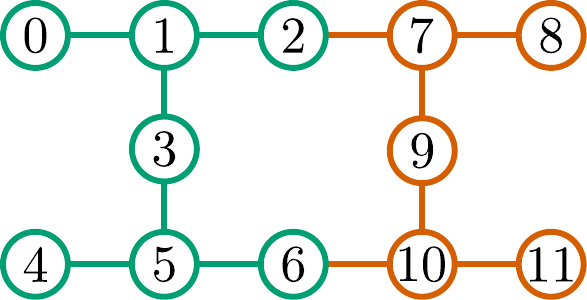}
			\caption{Topology of the simulated device.
				The qubits with index $0\ldots 6$ are based on the actual qubits of the \textit{ibmq\_perth}
				system, whereas the added qubits with index $7\ldots 11$ and their couplings are copies to increase
				the system size.}
			\label{fig:noisy_grid}
		\end{figure}
		
		\begin{table}[b]
			\centering
			\caption{Relevant metrics from the calibration snapshot of the \textit{imbq\_perth} system for
				our noise model. The table shows traverse (T1) and longitudinal (T2) relaxation times as well as
				singe-qubit gate times (t), single-qubit gate errors and read-out assignment errors for individual
				qubits of the system as well as two-qubit gate times and errors for each coupling. Calibration
				data was obtained on June 1\textsuperscript{st}, 2023.}
			\label{tab:noisy_cal}
			\adjustbox{width=0.5\textwidth}{
				\begin{tabular}{ccccccccc}
					\toprule
					\multicolumn{6}{c}{} & \multicolumn{3}{c}{\textbf{Couplings}} \\
					\cmidrule{7-9}
					\textbf{Qubit} &
					\textbf{T1 (\textmu s)} &
					\textbf{T2 (\textmu s)} &
					\textbf{t (ns)} &
					\textbf{Gate Err.} &
					\textbf{RO Err.} &
					\textbf{Cp.}   &
					\textbf{t (ns)} &
					\textbf{Gate Err.} \\
					\midrule
					0 & 9.6    & 16.47  & 35.56 & 0.0007 & 0.0220 & 0$\rightarrow$1: & 391.11 & 0.0129 \\
					1 & 150.65 & 55.92  & 35.56 & 0.0003 & 0.0232 & 1$\rightarrow$0: & 426.67 & 0.0129 \\
					&        &        &       &        &        & 1$\rightarrow$2: & 355.56 & 0.0052 \\
					&        &        &       &        &        & 1$\rightarrow$3: & 405.33 & 0.0145 \\
					2 & 120.61 & 103.28 & 35.56 & 0.0002 & 0.0205 & 2$\rightarrow$1: & 320.00 & 0.0052 \\
					3 & 169.23 & 151.85 & 35.56 & 0.0003 & 0.0176 & 3$\rightarrow$1: & 369.78 & 0.0145 \\
					&        &        &       &        &        & 3$\rightarrow$5: & 284.44 & 0.0086 \\
					4 & 159.29 & 117.10 & 35.56 & 0.0005 & 0.0178 & 4$\rightarrow$5: & 590.22 & 0.0103 \\
					5 & 187.23 & 140.95 & 35.56 & 0.0003 & 0.0240 & 5$\rightarrow$3: & 320.00 & 0.0086 \\
					&        &        &       &        &        & 5$\rightarrow$4: & 625.78 & 0.0103 \\
					&        &        &       &        &        & 5$\rightarrow$6: & 640.00 & 0.0102 \\
					6 & 163.37 & 180.04 & 35.56 & 0.0002 & 0.0060 & 6$\rightarrow$5: & 604.44 & 0.0102 \\
					\bottomrule
				\end{tabular}
			}
		\end{table}
		
		To simulate the quantum device noise, we assume decoherence is local and Markovian, which we model
		by thermal relaxation and depolarizing error channels~\cite{Wood2020} applied after each gate of the
		circuit.  We compose the thermal relaxation noise $\mathcal{E}_{TR}$ from an amplitude damping
		$\mathcal{E}_{AD}$ and a phase damping channel
		$\mathcal{E}_{PD}$~\cite{Georgopoulos2021,Nielsen2010}, such that for a quantum state $\rho$:
		\begin{equation}
			\mathcal{E}_{TR}(\rho) := \mathcal{E}_{PD}(\mathcal{E}_{AD}(\rho)).
		\end{equation}
		For simplicity and due to technical limitations of the software stack used for training the quantum
		models~\cite{Broughton2020TensorflowQuantum}, we omit the small contribution of the excited state
		population~\cite{Jin2015} to the thermal relaxation error and assume a qubit temperature of
		$\Theta=0K$. Hence, the amplitude damping channel is given by~\cite{Nielsen2010}:
		\begin{align}
			\mathcal{E}_{AD}(\rho) &\defeq K_{0}\rho K_{0}^{\dagger} + K_{1}\rho K_{1}^{\dagger}, \\
			K_{0} &\defeq
			\begin{bmatrix} 1 & 0 \\ 0 & \sqrt{1-\gamma_{AD}} \end{bmatrix}, \\
			K_{1} &\defeq \begin{bmatrix} 0 & \sqrt{\gamma_{AD}} \\ 0 & 0 \end{bmatrix},
		\end{align}
		with $\gamma_{AD}=1-e^{-\frac{t}{T_{1}}}$.
		The phase damping channel is defined accordingly as~\cite{Nielsen2010}:
		\begin{align}
			\mathcal{E}_{PD}(\rho) &\defeq K_{0}\rho K_{0}^{\dagger} + K_{1}\rho K_{1}^{\dagger},\\
			K_{0}&\defeq
			\begin{bmatrix} 1 & 0 \\ 0 & \sqrt{1-\gamma_{PD}} \end{bmatrix}, \\
			K_{1} &\defeq
			\begin{bmatrix} 0 & 0 \\ 0 & \sqrt{\gamma_{PD}} \end{bmatrix},
		\end{align}
		with $\gamma_{PD}=e^{-\frac{t}{T_{1}}} - e^{-\frac{2t}{T_{2}}}$.
		Here $t$ is the gate time, and $T_{1}$ and $T_{2}$ the traverse and longitudinal relaxation times.
		
		For the resulting error channel $\mathcal{E}_{TR}$, we compute the average fidelity~\cite{Nielsen2002, Gilchrist2005, Nielsen2010}
		\begin{equation}
			\mathcal{\bar{F}}(\mathcal{E}_{TR}) =
			\int\bra{\psi}\mathcal{E}_{TR}(\ket{\psi}\bra{\psi})\ket{\psi}d\psi
			= \frac{1}{2} + \frac{1}{6}e^{-\frac{t}{T_{1}}} + \frac{1}{3}e^{-\frac{t}{T_{2}}}.
		\end{equation}
		If the infidelity $1-\mathcal{\hat{F}}(\mathcal{E}_{TR})$ is smaller than the total gate error
		reported in the calibration data, we model the remaining infidelity with an additional depolarizing
		channel $\mathcal{E}_{DP}$~\cite{Nielsen2010}
		\begin{equation}
			\mathcal{E}_{DP}(\rho) := (1-p_{DP})\rho +
			\frac{p_{PD}}{3}\sum_{\hat{\sigma}\in\{\hat{X},\hat{Y},\hat{Z}\}}\hat{\sigma}\rho\hat{\sigma}.
		\end{equation}
		The depolarization probability $p_{DP}$ is derived from the gate error such that the average
		fidelity of the entire error channel $\mathcal{E}$ is
		\begin{align}
			\begin{split}
				\mathcal{\bar{F}}(\mathcal{E}) &= \mathcal{\bar{F}}(\mathcal{E}_{DP}\cdot\mathcal{E}_{TR}) \\
				&= (1-p_{DP})\mathcal{\bar{F}}(\mathcal{E}_{TR}) +
				p_{DP}\mathcal{\bar{F}}(\mathcal{E}_{D}),
			\end{split}
		\end{align}
		where $\mathcal{E}_{D}=\frac{\hat{I}}{2^{n}}$ is the completely depolarizing
		channel~\cite{Nielsen2010}.
		Since TensorFlow Quantum~\cite{Broughton2020TensorflowQuantum}, the quantum machine learning toolkit
		used in this work, does not simulate measurements on the circuit but determines expectation values
		for an observable directly from the computed quantum state, we approximate the readout error of the
		simulated device by appending a bit-flip channel $\mathcal{E}_{BF}$~\cite{Nielsen2010} to the end of
		the circuit. The channel is defined by
		\begin{equation}
			\mathcal{E}_{BF}(\rho) := (1-p_{BF})\rho + p_{BF}\hat{X}\rho\hat{X}.
		\end{equation}
		The probability for the bit-flip $p_{BF}$ is set to the readout assignment error probability from
		the calibration data for each qubit.
		
		\subsection{Evaluation and results}\label{ap:noisy_shots}
		
		We provide an outline on how to use the learning capability in the presence of noise.
		Although a more detailed investigation is out of scope of this manuscript, this section shows that the learning capability can, in principle, be evaluated with a quantum circuit simulator that incorporates noise as described in the previous Appendix~\ref{ap:noise_model}.
		At first, normalizing the Fourier function values to values less than one allows the PQCs to compensate for the noise effects.
		We do so by dividing each function value by $\frac{3}{4}$.
		Furthermore, to reduce the computational effort, we restrict the set of Fourier series to $10$ functions and select them by choosing the ones that are included in the function pairs with the lowest cross-correlation values as calculated in Appendix~\ref{ap:fourier}.
		
		\begin{figure}[t]
			\centering
			\includegraphics[width=0.75\textwidth]{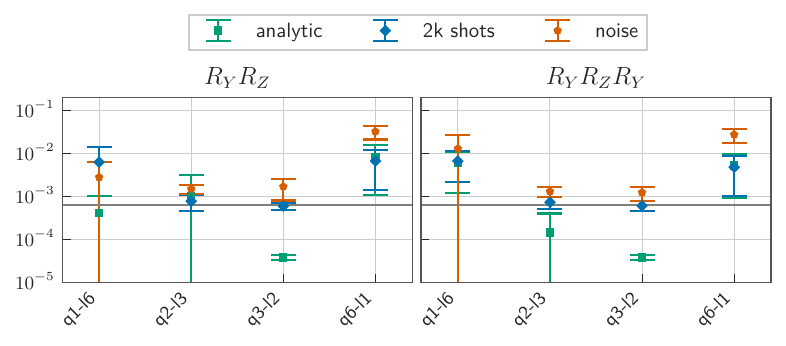}
			\caption{Comparison of $\mu_6$ for different layered ans\"atze with two simple, linear $\CNOT$
				entanglement layers.  The left plot shows results for $R_Y R_Z$, while the right plot shows
				results for $R_Y R_ZR_Y$ single-qubit operations.  The green data points show analytic
				simulation results as in Fig.~\ref{fig:pqc_d6_cz_urot}.  The blue data points depict results
				of shot-based expectation values using $2000$ shots and the parameter shift rule.  The orange
				data points are based on shot-based expectation values for $2000$ shots and the noise model
				discussed in the previous Appendix~\ref{ap:noise_model}.} \label{fig:noise_norm_75}
		\end{figure}
		
		We test the experiments for degree $6$ functions and ans\"atze with two simple, linear $\CNOT$ entanglement layers.
		Figure~\ref{fig:noise_norm_75} shows the results for these experiments for analytical simulation, shot-based, and noisy shot-based expectation values with $2000$ shots.
		We select qubit $\{0\}$ for 1 qubit, $\{0,1\}$ for 2 qubits, $\{0,1,2\}$ for 3 qubits and $\{0,1,2,7,9,10\}$ for 6 qubits from the approximated processor in Fig.~\ref{fig:noisy_grid}.
		
		The analytical case shows that the changes to reduce the computational effort
		still enable the reproduction of the characteristic curve
		except for the one qubit, six layers case with $R_YR_Z$ single-qubit gates
		which performs better than expected.
		The shot-based experiments without noise, however,
		follow the exact characteristic curve with a higher minimal loss for all runs
		as expected from Appendix~\ref{ap:shots}.
		That is, the two and three qubits architectures perform the best and the one and six qubits architectures lead to worse results.
		Including noise weakens the learning capability further.
		However, the characteristic curve still remains and shows differences in the choice of qubits and layers.
		Again, the one qubit case performs slightly better than expected which could be due to statistical variances or because of the absence of two-qubit errors.
		Further noise experiments would be necessary to investigate this finding in more detail.
		
		\section{Hyperparameters} \label{hyper}
		Our definition of learning capability depends on hyperparameters of supervised learning, namely the optimizer, the initialization, epochs, and learning rates. On the one hand, this results from the fact that we provide a measure close to practical work. On the other hand, this could lead to different results depending on the hyperparameters and not solely on the chosen ansatz. Therefore, we provide a comparison of results obtained by different hyperparameters in Fig.~\ref{fig:hyperparameters}. The data in the figure shows the following:
		\begin{enumerate}
			\item Losses are very stable in the range of the tested hyperparameters for the tested ans\"atze.
			\item The gap between one and two entanglement layers in Fig.~\ref{fig:coeffs_funct} is not resolved by any of these hyperparameter sets.
			\item The gap between dQNNs with and without reupload on hidden neurons in Fig.~\ref{fig:dQNN_dr_zl} cannot be resolved.
		\end{enumerate}
		This grid test of hyperparameters gives strong evidence that differences in learning capabilities can be explained by differences in the architecture. However, since our grid test could still miss an important hyperparameter set, we conclude that these sets would be hard to find and thus the ansatz were more impractical then others, because extensive hyperparameter optimization would be needed to obtain similar results.
		
		\begin{figure}[t]
			\centering
			\includegraphics[width=0.75\textwidth]{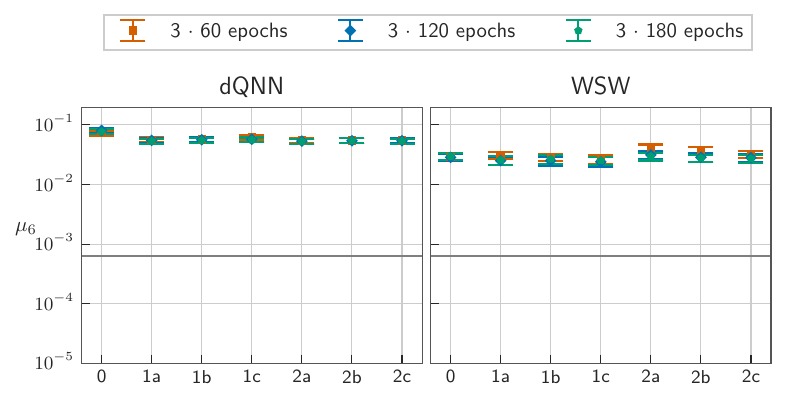}
			\caption{Overview of values for $\mu_6$ for different hyperparameter sets. The results are depicted for a layered ansatz with $q=3$, $l=2$, $\text{el}=1$, $U^1=R_y$ and $\CRX$ and a dQNN ansatz with $[61]$ and $U^1 = R_YR_ZR_Y$. The values on the $x$ axis are explained in Table~\ref{tab:hyperopt_pqc}. The legend shows the different epoch sizes per learning rate.
				For example, the blue marker for setting 2a corresponds to learning $120$ epochs with learning rate $0.1$, $120$ epochs with learning rate $0.05$, and finally $120$ epochs with learning rate $0.01$.}
			\label{fig:hyperparameters}
		\end{figure}
		
		\begin{table}[b]
			\centering
			\ra{1.2}
			\caption{Learning rates for results in Fig.~\ref{fig:hyperparameters}.}
			\label{tab:hyperopt_pqc}
			\begin{tabular}{@{}rcrrr@{}}
				\toprule
				&& \multicolumn{3}{r}{Learning rates}\\
				Set & \phantom{a} & $l_0$ & $l_1$ & $l_2$ \\
				\midrule
				$0$ && $0.3$ & $0.3$ & $0.3$\\
				1a && 0.5 &  0.1 & 0.05 \\
				1b && 0.5 &  0.1 & 0.1 \\
				1c && 0.5 &  0.5 & 0.1 \\
				2a && 0.1 & 0.05 & 0.01 \\
				2b && 0.1 & 0.05 & 0.05 \\
				2c && 0.1 &  0.1 & 0.05 \\
				\bottomrule
			\end{tabular}
		\end{table}
		
		\section{Details on the sets of Fourier functions} \label{ap:fourier}
		
		To generate truncated random Fourier series of degree $d$, we sample for each frequency $\omega$ a complex-valued coefficient $c_\omega = a_\omega + i b_\omega$ as two random numbers between $-0.5$ and $0.5$ using a uniform distribution.
		Due to the constraints of PQCs, the coefficients corresponding to $-\omega$ are given by the complex conjugate $c_{-\omega} = \cc{c_\omega}$.
		Subsequently, all coefficients are normalized such that the highest absolute value of the truncated Fourier series is equal to 1.
		The code to generate the set of random truncated Fourier series is given in  \href{https://github.com/dfki-ric-quantum/learning_capability/blob/master/create_rnd_functions.ipynb}{Notebook}.
		
		As it was explained in Sec.~\ref{sec:lc}, the mean error $\mu_d$ is given by
		\begin{equation}
			\mu_d = \frac{1}{|G_d|} \sum_{g \in G_d} \epsilon_g,
		\end{equation}
		where $\epsilon_g$ is the final validation loss for the Fourier series $g$. Each parameterized quantum circuit ansatz is reinitialized and trained on a sample for several truncated Fourier series which form the set $G_d$.
		For very similar functions in the set one might get the illusion that a model performed well on $n$ Fourier series, when in fact several of them are the same.
		Therefore, a key point is to ensure that the sampled Fourier series are different to guarantee that we minimize the effect of the sample on the results.
		This is done by computing the cross-correlation matrix of all the Fourier series in the set $G_d$.
		
		The cross-correlation is a similarity measure between two series which can be used to objectively quantify how similar two (or more) series are. For two continuous periodic functions $f$ and $g$ with period $T$, the cross correlation is
		\begin{align}
			(f*g)(\tau) = \int_{t_0}^{t_0+T}\overline{f(t)}g(t+\tau)dt
		\end{align}
		where $\overline{f(t)}$ is the complex conjugate of $f(t)$ and $\tau$ is the lag.
		For the $\tau$ values, we divide the interval $[0, 2\pi]$ into $100$ equal parts.
		We calculate the cross-correlation for $\tau = \frac{2\pi \cdot k}{100} \text{ with } k \in \{0, \ldots, 99\}$.
		Clearly, for each value of the lag we get a value of the cross-correlation of $f$ and $g$.
		
		In our case, we are interested in the largest value this cross-correlation could take between the two Fourier series, meaning the lag for which the two Fourier series are most similar. This way, we know that the two Fourier series are not the same up to a shift in the $x$ axis.
		
		The analysis of the cross correlation of our Fourier series data set of degree $12$ is given in Table~\ref{tab:cross_cor}.
		For two data sets, $G_{12}$ and $G^\prime_{12}$, the absolute values of the upper triangle matrix without diagonal elements are listed by counting the number values lying in the specified interval.
		The values in the table show that the Fourier functions are  in general uncorrelated.
		The function pairs with the lowest and highest cross-correlation are depicted in Fig.~\ref{fig:cross_corr} which also visualizes the differences between the original and shifted functions.
		
		\begin{table}[b]
			\centering
			\ra{1.2}
			\caption{Upper triangle of the cross-correlation matrix.}
			\label{tab:cross_cor}
			\begin{tabular}{rrrrrrrrrrr}
				\toprule
				& $[0.0,$ & $[0.1,$ &  $[0.2,$ &  $[0.3,$  &  $[0.4,$  &  $[0.5,$ &  $[0.6,$ &  $[0.7,$ &  $[0.8,$ &  $[0.9,$\\
				set & $0.1)$ & $0.2)$ &  $0.3)$ &  $0.4)$  &  $0.5)$  &  $0.6)$ &  $0.7)$ &  $0.8)$ &  $0.9)$ &  $1.0]$\\
				\midrule
				$G_{12}$ & 0 & 0 & 0 & 254 & 2350 & 1925 & 383 & 38 & 0 & 0\\
				$G^\prime_{12}$ & 0 & 0 & 0 & 268 & 2404 & 1823 & 412 & 43 & 0 & 0\\
				\bottomrule
			\end{tabular}
		\end{table}
		
		\begin{figure}[t]
			\centering
			\begin{subfigure}[H]{0.475\textwidth}
				\centering
				\includegraphics[width=\textwidth]{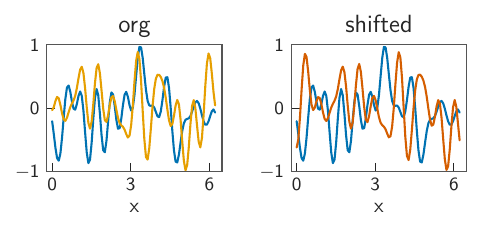}
				\caption{Best cross-correlation function pairs.}
			\end{subfigure}
			\begin{subfigure}[H]{0.475\textwidth}
				\centering
				\includegraphics[width=\textwidth]{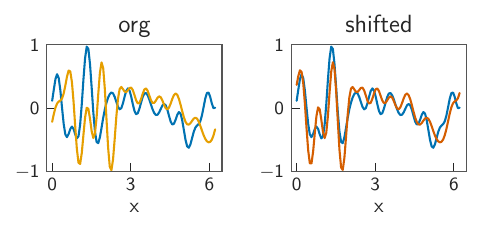}
				\caption{Worst cross-correlation function pairs.}
			\end{subfigure}
			\caption{
				Cross-correlation for the original dataset of degree 12 functions.
				The function pair with the lowest and the highest correlation is depicted in blue and red.
				Because we consider shifted curves to determine the maximal cross-correlation value for a function pair, the original curve is depicted in yellow.}
			\label{fig:cross_corr}
		\end{figure}
		
		\begin{figure*}[t]
			\centering
			\includegraphics[width=0.75\textwidth]{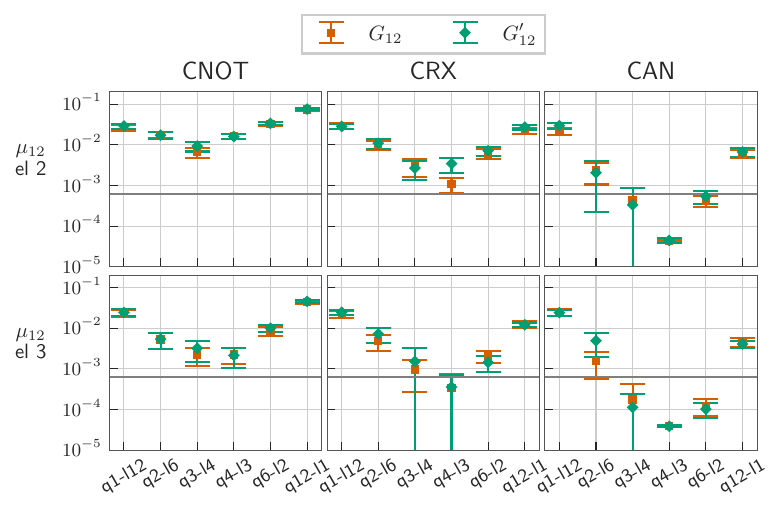}
			\caption{Comparison of learning capabilities for two different sets of randomly sampled truncated Fourier series. Orange data points show results for the original data as in Fig.~\ref{fig:pqc_d12_all}. Green data points show the result of running the experiments on a different set of truncated Fourier series. The learning capability $\mu_{12}$ is for different $WSW$ ans\"atze with $R_YR_Z$ single-qubit gates, simple linear entangling structure and $\{\CNOT, \CRX, \CAN\}$ entangling gates. The top row shows results for two entanglement layers, while the bottom row shows results for three entanglement layers.} \label{fig:wsw_d12_fourier}
		\end{figure*}
		
		To test the impact of different sets $G_{12}, G^\prime_{12}$
		We compare the results for the learning capability for two different sets of randomly sampled truncated Fourier series of degree $12$ in Fig.~\ref{fig:wsw_d12_fourier}.
		This shows that there is a difference in the results for the two sets, but that this difference is rather marginal and not statistically significant.
		
		\section{Preliminaries}
		\label{Preliminaries}
		
		We present introductory results for simple circuits which underline that data reupload can in principle provide a frequency spectrum of size $K =n L$ but that the structure of the trainable unitaries $W$ determines whether all Fourier coefficients can indeed be represented.
		The circuits for this section are shown in Appendix~\ref{circuits}.
		This Appendix reproduces the results from Ref.~\cite{Schuld2021} for simple examples and demonstrates that the learning capability naturally captures the numerical results.
		
		\subsection{One qubit and one layer} \label{subsec:warm_up_1q1l}
		
		\begin{figure}[tb]
			\centering
			\includegraphics[width=0.5\linewidth]{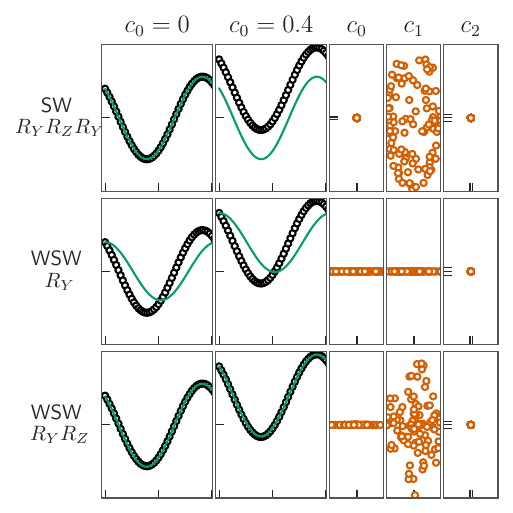}
			\caption{
				\textbf{Left columns}: Validation results (green lines) of two different trained functions (black circles) with different $c_0 \in \{0,0.4\}$.
				The x axis is from $0$ to $2\pi$, the y axis from $-1$ to $1$.
				\textbf{Right columns}:
				Fourier coefficients $c_0$ to $c_2$ resulting from the inverse Fourier transform of the results of evaluating the three ans\"atze with sampled parameters $\Theta$.
				The x axis shows the real part of the coefficient, while the y axis shows the imaginary part. Both axis are from $-1$ to $1$.
				\newline
				\textbf{First row}: For circuits with depth one - one qubit and only one building block $SW$-, the Fourier coefficient $c_0$ is always zero and corresponding functions cannot be fully approximated (see also the calculation in Example~\ref{ex:sw_1q1l}).
				\newline
				\textbf{Second row}: A $WSW$ structure with depth one and one qubit with $R_Y$ as the rotation gate leads to real Fourier coefficients (see also Example~\ref{ex:wsw_1q1l}). \newline
				\textbf{Third row}: Once we use a $WSW$ structure with depth one on one qubit and $R_YR_Z$ rotation gates, $c_0 \neq 0$ is reachable and $\abs{\Im{c_1}}>0$, so that corresponding Fourier functions can be fitted.}
			\label{fig:1q1l_sw_wsw}
		\end{figure}
		
		Some frequently used circuits start with the input encoding directly on the input qubits without an initial block of parametrised unitary gates (see for example Refs.~\cite{Matic2022, Mangini2022, Moussa2022}). For a systematic comparison we refer to this idea as $SW$ instead of $WSW$.
		If we consider the elementary case of one qubit and one layer, then a calculation shows that $SW$ has not enough expressiveness to fit all Fourier functions with degree 1 because the coefficient $c_0$ of Eq.~\eqref{eq:fx_singlequbit_singlelayer} is always 0.
		The calculation is given in Example~\ref{ex:sw_1q1l} in Appendix~\ref{ap:calculations}.
		It uses a $\sigma_z$ measurement but is independent of the exact implementation of the unitary $U^1$.
		This fact can be reproduced numerically in two different ways by choosing three different ans\"atze:
		The first ansatz is based on a $SW$ structure with $U^1=R_YR_ZR_Y$, i.e. $U_{SW,3}(x,\Theta) =  R_Y(\theta_{13}) R_Z(\theta_{12}) R_Y(\theta_{11}) R_X(x)$;
		the second one has a zero layer but only one gate as the single-qubit operation:
		$U_{WSW, 1}(x, \Theta) = R_Y(\theta_{11}) R_X(x) R_Y(\theta_{01})$;
		and the third ansatz consists of a zero layer and two gates per single-qubit operation:
		$U_{WSW, 2}(x, \Theta) = R_Z(\theta_{12}) R_Y(\theta_{11}) R_X(x) R_Z(\theta_{02}) R_Y(\theta_{01})$.
		Each ansatz corresponds to a row in Fig.~\ref{fig:1q1l_sw_wsw} and the circuits are depicted in Appendix~\ref{circuits}.
		
		Two tools, introduced in Ref.~\cite{Schuld2021}, can be used to analyze the expressive power numerically:
		The first tool is to train the ans\"atze to minimize the loss for a chosen Fourier functions of degree 1.
		In Fig.~\ref{fig:1q1l_sw_wsw} the results are depicted in the first two columns for two different functions. The first function is determined by $c_0 = 0$ and $c_1 = 0.2 + i0.2$ while the second function is given by $c_0 = 0.4$ and $c_1 = 0.2 + i0.2$.
		This way, both functions are normalized such that no post-processing is necessary.
		These numerical experiments confirm the calculation
		because the first row shows that $U_{SW,3}(\boldsymbol{x,\theta})$ does only provide $c_0=0$.
		The ansatz $U_{WSW,1}(x,\boldsymbol{\theta})$ (second row) is not able to exactly fit the function while
		the ansatz $U_{WSW,2}(x,\boldsymbol{\theta})$ (last row) results in perfect fits.
		The second tool to analyze the expressive power involves randomly initializing the circuits multiple times; by performing an inverse Fourier transformation, the imaginary and real values for each Fourier coefficient can be obtained. Each initialization is represented by an orange circle in the three last columns of Fig.~\ref{fig:1q1l_sw_wsw} where the third last column shows $c_0$, the second last $c_1$, and the last column $c_2$.
		Imaginary values are plotted on the $y$ axis, whereas real values are plotted on the $x$ axis.
		Note that $c_0$ can only be real because of Eq.~\eqref{eq:fx_pqc_model}.
		Again, $U_{SW,3}(x,\boldsymbol{\theta})$ misses $c_0 \neq 0$,
		$U_{WSW, 1}(x,\boldsymbol{\theta})$ only provides  real values for both reachable Fourier coefficients,
		and $U_{WSW, 2}(x,\boldsymbol{\theta})$ enables both coefficients.
		
		\subsection{Minimum number of qubits and layers}\label{Preliminaries_Multi}
		
		\begin{figure*}[t]
			\centering
			\includegraphics[width=0.95\linewidth]{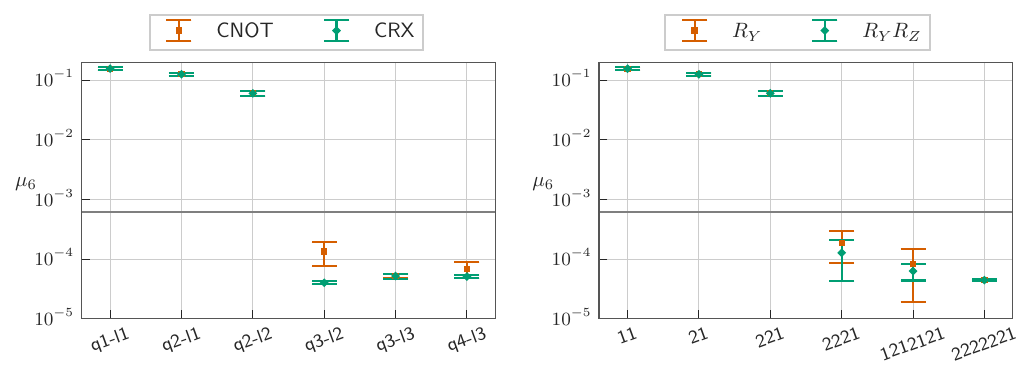}
			\caption{Comparison of $\mu_6$ for layered (left) and dQNN (right) architectures of different layer and qubit numbers. \newline
				\textbf{Left}: results for layered ans\"atze with $R_YR_Z$ single-qubit gates and $\{\CNOT, \CRX\}$ entangling gates for qubit numbers 1 to 4, layer numbers 1 to 3, and two entanglement layers per $W$; \textbf{Right}: results for dQNN ans\"atze with $R_Y$ and $R_YR_Z$ single-qubit gates and different numbers of qubits in (hidden) layers on the $x$ axis.\newline
				The plots show that for these combinations of single-and two-qubit gates, which are known to perform well from the results shown in the main text, the optimal learning capability is reached once the number of qubits and layers matches the degree. Note that the results for $\mu_6$ for the first three layer and qubit combinations on the left and right are nearly indistinguishable on the logarithmic scale.} \label{fig_d6_compare}
		\end{figure*}
		
		According to the results in Ref.~\cite{Schuld2021}, we
		are also able to capture the fact that every ansatz needs a certain number of layers $L$ and qubits $n$ to be able to, in principle, reach a sufficient learning capability for a certain Fourier degree $d_{\text{max}}$, with $d_{\text{max}} = L n$.
		The corresponding numerical results are shown in Fig.~\ref{fig_d6_compare}, where we fit 100 random normalized Fourier functions of degree 6 with different ans\"atze to determine $\mu_6$ for each ansatz.
		The left side of Fig.~\ref{fig_d6_compare} shows two layered ans\"atze:
		Both utilize $U^1=R_YR_Z$ and a simple, linear entanglement structure with two entanglement layers but different entanglement gates $\{\CNOT, \CRX\}$.
		The number of qubits in these architectures is increased from 1 to 4 and the number of layers from 1 to 3.
		The learning capability together with its $95\%$ confidence interval is plotted in each figure as circles with error bars.
		Once  $nL = d=6$, the learning capability is enhanced by more than two orders of magnitude. However, it is not improved significantly for ans\"atze with larger $L$ or $n$.
		
		The same holds for modified dQNN ans\"atze as depicted in Fig.~\ref{fig_d6_compare} on the right side where both ans\"atze use circuits of the form in Fig.~\ref{fig:dqnn_dr_zl} but vary in the single-qubit operation $\{R_Y, R_YR_Z\}$.
		The architectures increase the number of input qubits from one to two and the hidden qubits from 0 to 10 (which are split in different layers for some architectures).
		For example, the notation $221$ represents two input and one output qubit, with a hidden layer with two qubits in between.

		\section{Analysis of error distribution} \label{ap:error}
		
		We calculate the mean $\mu$ and its confidence interval $[\mu-I, \mu+I]$ with
		$I = c\cdot\widehat{\sigma}_N$ by assuming a student's $t$ distribution with $N=100$ samples,
		$c = 1.98$ and use the standard error of the mean (SEM) $\widehat{\sigma}_N$
		given by $\widehat{\sigma}_N = \frac{\sigma_N}{\sqrt{N}}$
		where $\sigma_N = \sqrt{\frac{\sum_{i=1}^N (X_i - \mu)^2}{N-1}} $
		is the corrected standard deviation of the sample
		with $\mu$ being the sample's mean
		over all final MSEs $\{X_i\}_{i=0}^N$ obtained for each function in the set.
		
		\begin{figure}[t]
			\centering
			\includegraphics[width=\textwidth]{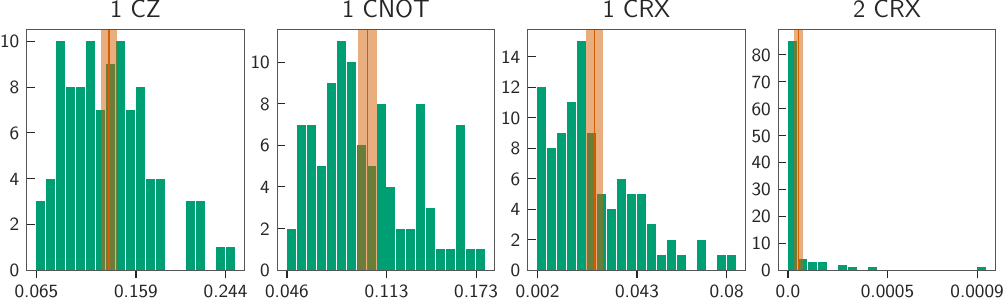}
			\caption{Detailed error distribution of validation error results in Fig.~\ref{fig:coeffs_funct}. For each ansatz, the range of all $100$ final validation errors is divided into $20$ equal baskets.
				The basket is determined by the minimum and maximum error for each configuration
				individually resulting in different $x$ axis.
				The $y$ axis counts the number of errors included in each basket. The mean with its $95\%$ confidence interval is depicted as a red vertical line.} \label{fig:error_bardistribution}
		\end{figure}
		
		To further motivate the mean in Eq.~\eqref{eq:learning_capability} from an empirical perspective, we analyze the errors of the runs depicted in Fig.~\ref{fig:coeffs_funct}
		of the result Sec.~\ref{subsec:schuld}. Figure~\ref{fig:error_bardistribution} shows the error distribution together with its mean and $95\%$ confidence interval.
		
		\begin{figure}[t]
			\centering
			\includegraphics[width=\textwidth]{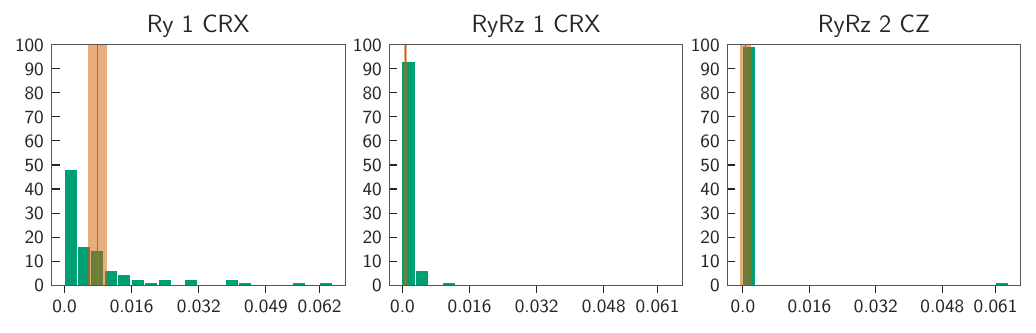}
			\caption{Detailed error distribution of validation MSE results in Fig.~\ref{fig:d6_alt} for selected configurations. For each ansatz, the range of all $100$ final validation errors is divided into $20$ equal baskets.
				The baskets are the same for all three configurations and are determined by their
				total minimum and maximal error.
				The $y$ axis counts the number of errors included in each basket. The mean with its $95\%$ confidence interval is depicted as an orange vertical line.} \label{fig:error_bardistribution_q3_l2_alt}
		\end{figure}

		In Fig.~\ref{fig:d6_alt}, the configuration $R_YR_Z$ 2 $\CZ$
		seems to lead to a large confidence interval compared to the other configurations.
		For this configuration, the mean is $6.8\times 10^{-4}$
		and the SEM is $6.4\times 10^{-4}$, which is larger than that
		for $R_YR_Z$ 1 $\CRX$ ($1.5\times 10^{-4}$)
		because an outlier result at $> 6.1\times 10^{-2}$ increases the variance.
		The error distribution is depicted in Fig.~\ref{fig:error_bardistribution_q3_l2_alt}.
		Furthermore, the logarithmic scale in the result plots and cutting the plots at $10^{-5}$
		lead to seemingly skewed confidence intervals.
		In contrast, for example, the $R_Y$ 1 $\CRX$ configuration has a higher SEM
		with $1.2\times 10^{-3}$.
		
		\begin{figure}[t]
			\centering
			\includegraphics[width=\textwidth]{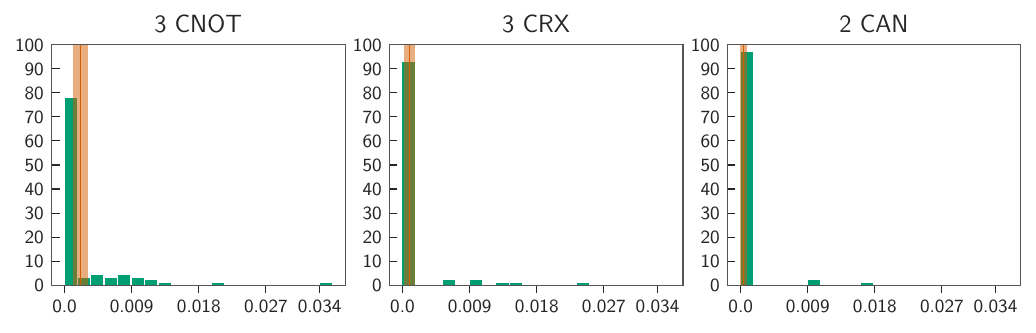}
			\caption{Detailed error distribution of validation MSE results in
				Fig.~\ref{fig:loss_degree12_entgates} for three qubits and four layers
				configurations.
				For each ansatz, the range of all $100$ final validation errors is divided into $20$ equal baskets.
				The baskets are the same for all three configurations and are determined by their
				total minimum and maximal error.
				The $y$ axis counts the number of errors included in each basket. The mean with its $95\%$ confidence interval is depicted as an orange vertical line.} \label{fig:error_bardistribution_q3_l4}
		\end{figure}
		
		\begin{figure}[t]
			\centering
			\includegraphics[width=\textwidth]{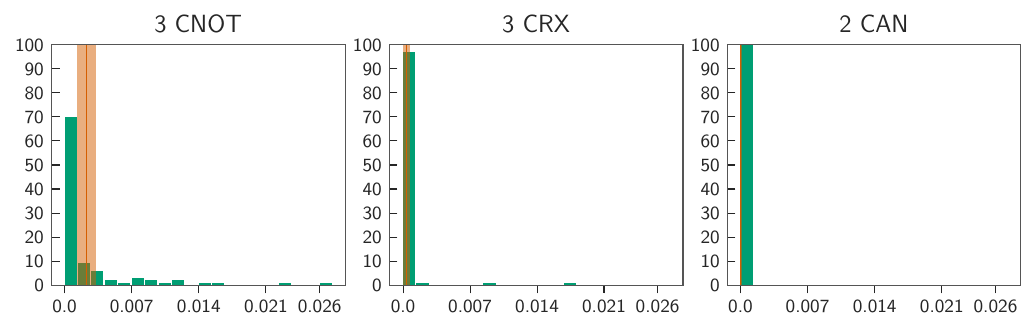}
			\caption{Detailed error distribution of validation MSE results in
				Fig.~\ref{fig:loss_degree12_entgates} for four qubits and three layers
				configurations.
				For each ansatz, the range of all $100$ final validation errors is divided into $20$ equal baskets.
				The baskets are the same for all three configurations and are determined by their
				total minimum and maximal error.
				The $y$ axis counts the number of errors included in each basket. The mean with its $95\%$ confidence interval is depicted as an orange vertical line.} \label{fig:error_bardistribution_q4_l3}
		\end{figure}
		
		We provide two more comparisons of seemingly large confidence intervals from Fig.~\ref{fig:loss_degree12_entgates} for degree $12$ Fourier functions.
		In Fig.~\ref{fig:error_bardistribution_q3_l4}, respectively, Fig.~\ref{fig:error_bardistribution_q4_l3}, the results for configurations containing
		three qubits and four layers, respectively, four qubits and three layers are depicted.
		In both cases, the configurations with errors in lower baskets have lower means
		and smaller confidence intervals because of less outlier which underlines that
		logarithmic scaling is one major reason for the seemingly large confidence intervals
		for those cases.
		
		\section{Analysis of the size of the underlying set of Fourier functions} \label{ap:error_sampling}
		We investigate the impact of the data set's size on the mean and confidence interval
		of our result plots. For the same configuration examples as in the previous Appendix,
		we create subsets by uniformly randomly removing five elements from the results
		to obtain subsets of sizes from $5$ to $100$ in steps of five.
		Calculating the mean and confidence interval for each subset enables us to analyze
		the dependence of the number of functions in $G_d$ on the learning capability value.
		
		\begin{figure}[t]
			\centering
			\includegraphics[width=\textwidth]{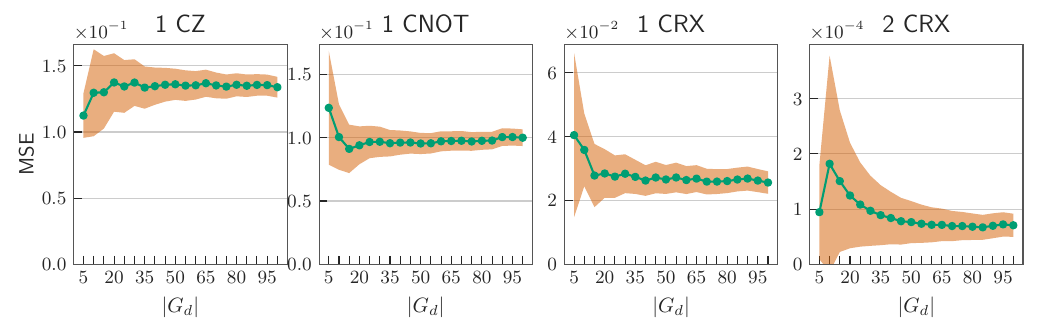}
			\caption{Mean and confidence interval of validation MSE results for successive
				smaller set sizes $\abs{G_d}$ for three qubits and four layers
				configurations of Fig.~\ref{fig:coeffs_funct}.
				The plot shows a smooth convergence of the mean for all four configurations
				for set sizes $\gtrsim 25$.} \label{fig:sampled_mean_q3_l2_lin}
		\end{figure}
		\begin{figure}[t]
			\centering
			\includegraphics[width=\textwidth]{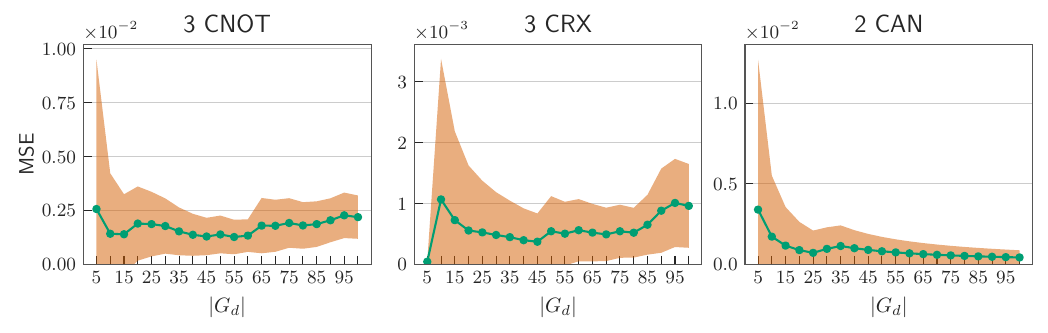}
			\caption{Mean and confidence interval of validation MSE results for successive
				smaller set sizes $\abs{G_d}$ for three qubits and four layers
				configurations of Fig.~\ref{fig:loss_degree12_entgates}.
				The plot shows a smooth convergence for $2 \CAN$.
				For both other cases  the convergence of the mean is less smooth which
				suggests that the size $\abs{G_d} = 100$ should not be decreased.}
			\label{fig:sampled_mean_q3_l4}
		\end{figure}
		\begin{figure}[t]
			\includegraphics[width=\textwidth]{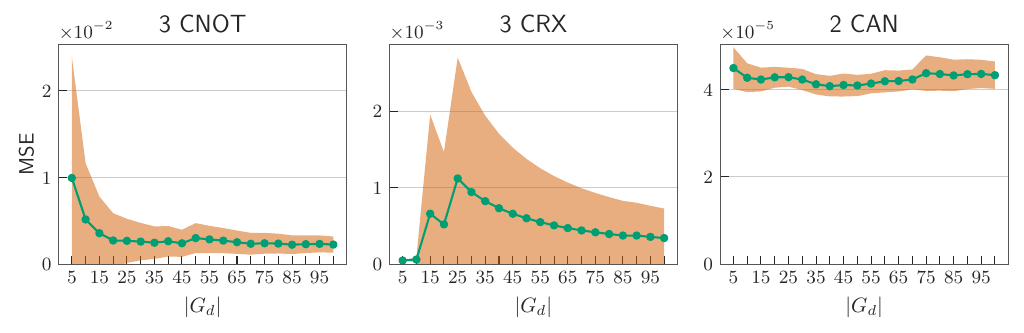}
			\caption{Mean and confidence interval of validation MSE results for successive
				smaller set sizes $\abs{G_d}$ for four qubits and three layers
				configurations of Fig.~\ref{fig:loss_degree12_entgates}.
				The plot shows a smooth convergence for all three configurations.
				However, the shape of the curve for $3 \CRX$
				suggests that the size $\abs{G_d} = 100$ should not be decreased.}
			\label{fig:sampled_mean_q4_l3}
		\end{figure}
		
		The data leading to the results of Fig.~\ref{fig:coeffs_funct} is analyzed in
		Fig~\ref{fig:sampled_mean_q3_l2_lin} showing a smooth convergence towards the
		reported mean for set sizes $\gtrsim 25$.
		In addition, the data for three qubits and four layers configurations, respectively, four qubits and three qubits, of Fig.~\ref{fig:loss_degree12_entgates}
		are analyzed in Fig.~\ref{fig:sampled_mean_q3_l4}, respectively, Fig.~\ref{fig:sampled_mean_q4_l3}.
		The convergence for the $\CRX$ gate with three entanglement layers is less smooth than
		for the others which suggests that the size of $\abs{G_d} = 100$ should not be decreased.
		The other plots show a smooth convergence to the mean.

		\section{Barren plateaus} \label{barren}
		
		We numerically check if barren plateaus can explain the results of the layered PQCs  depicted as blue points in Fig.~\ref{fig:loss_degree12_entgates} and the blue points for dQNNs in Fig.~\ref{fig:dQNN_dr_zl} (or Fig.~\ref{fig:dQNN_urot}).
		The layered WSW ans\"atze with $n\cdot L=12$ vary in their number of qubits $n$ and layers $L$ but all consist of single-qubit rotations $R_YR_Z$ and $3$ entanglement layers in a simple, linear structure utilizing $\CRX$ entanglement gates.
		The dQNNs ans\"atze enabling degree $6$ are the following (in the notation of Eq.~\eqref{eq:notation_qubits_dqnn}):
		\begin{itemize}
			\item $[6,1]$ (max.~$7$ neighboring qubits),
			\item $[3,3,1]$ (max.~$6$ neighboring qubits),
			\item $[2,2,2,1]$ (max.~$4$ neighboring qubits),
			\item $[1,2,1,2,1]$ (max.~$3$ neighboring qubits),
			\item $[1,1,1,1,1,1,1]$ (max.~$2$ neighboring qubits).
		\end{itemize}
		All dQNNs utilize single-qubit operations $R_YR_Z$ and a data reupload scheme with a zero layer on each qubit.
		
		Following the approach of Ref.~\cite{McClean2018}, we initialize every trainable gate randomly in the interval $[0,2\pi)$ except for the first gate acting on the first qubit.
		We calculate the gradient based on the mean squared error loss of the complete training data set of 50, respectively, 100, data points for degree $6$, respectively, degree $12$, Fourier series and determine the variance of these gradients by averaging over the $100$ Fourier functions used to calculate the learning capability.
		
		In the left plot of Fig.~\ref{fig:barren}, the results for dQNNs are depicted.
		No tendency for different choices of hidden qubits is observable indicating that barren plateaus cannot explain the drastically different performances represent by the blue points in Fig.~\ref{fig:dQNN_dr_zl} (or Fig.~\ref{fig:dQNN_urot}).
		
		In the right plot of Fig.~\ref{fig:barren}, the results for layered ans\"atze are depicted.
		A clear tendency, in agreement with the literature, can be found showing an exponential decay of the variance of the gradient when increasing the number of qubits (and decreasing the number of layers).
		However, this does not adequately explain the blue curve in Fig.~\ref{fig:loss_degree12_entgates}, because the architecture with four qubits (three layers) performs slightly better than the one with three qubits (four layers); the ansatz with six qubits (two layers) performs slightly better than the one with two qubits (six layers); and the ansatz with $12$ qubits (one layer) performs slightly better than the one with one qubit ($12$ layers).
		
		This numerical analysis further supports our claim that the effect of barren plateaus is not sufficient to explain different learning capabilities of different ans\"atze for dQNNs and layered PQCs.
		
		\begin{figure}[t]
			\centering
			\includegraphics[width=\textwidth]{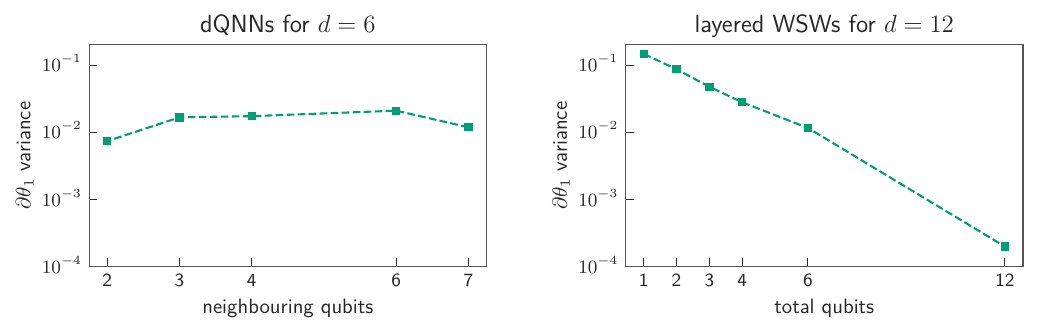}
			\caption{Investigation of possible barren plateaus via the variance of parameters' gradients~\cite{McClean2018}. Except for the first gate acting on the first qubit, all gates of the circuit are initialized randomly. The gradient is calculated based on the mean squared error loss of the complete training data set.
				\textbf{Left}: Results for dQNN ans\"atze enabling degree $6$ Fourier series corresponding to the ans\"atze resulting in the blue curve in Fig.~\ref{fig:dQNN_dr_zl} and Fig.~\ref{fig:dQNN_urot}.
				\textbf{Right}: Results for layered ans\"atze for degree $12$ Fourier series corresponding to the blue curve in Fig.~\ref{fig:loss_degree12_entgates}.}
			\label{fig:barren}
		\end{figure}
		
		\section{Additional figures and tables}
		\label{Appendix_plots}
		
		In this appendix, we provide two more figures and tables. The figures show detailed comparisons of layered ansätze with different entanglement layers and styles for degree $12$ in Fig.~\ref{fig:pqc_d12_all} and for degree $6$ in Fig.~\ref{fig:pqc_d6_cz_urot}.
		
		Table~\ref{tab:parameter_count_WSW_d12} and Table~\ref{tab:parameter_count_WSW_dQNN} list the number of single-qubit gates, two-qubit gates, and trainable parameters for the ans\"atze in Fig.~\ref{fig:loss_degree12_entgates} and Fig.~\ref{fig_d6_compare}, respectively.
		
		\begin{figure*}[t]
			\centering
			\includegraphics[width=\textwidth]{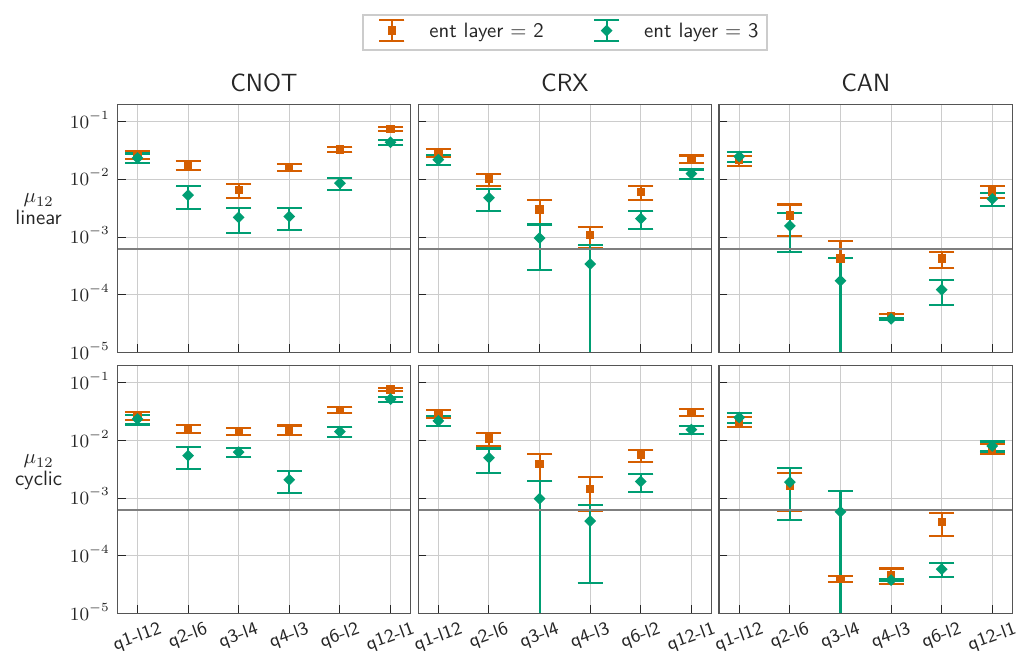}
			\caption{Results for the learning capability $\mu_{12}$ for different $WSW$ ans\"atze with $R_YR_Z$ single-qubit gates, simple entangling structure and $\{\CNOT, \CRX, \CAN\}$ entangling gates. The top row shows results for linear entanglement style, while the bottom row shows results for cyclic entanglement (see Fig.~\ref{fig:qvc_ent_style} for an explanation of the entanglement styles).} \label{fig:pqc_d12_all}
		\end{figure*}
		
		\begin{figure*}[t]
			\centering
			\includegraphics[width=0.85\textwidth]{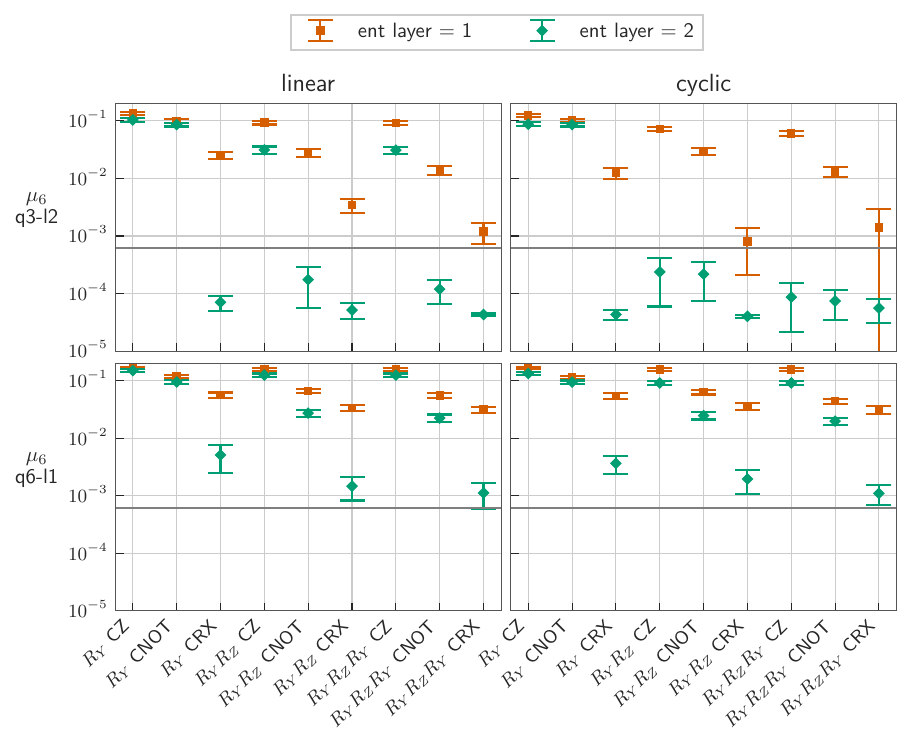}
			\caption{Comparison of $\mu_6$ for different layered ans\"atze with simple and linear (left) or cyclic (right) entanglement structure. The top rows shows results for three qubits and two layers, while the bottom row shows results for six qubits and one layer. The number of entanglement layers $\{1, 2\}$ is defined as shown in Fig.~\ref{fig:qvc_ent_layer}. } \label{fig:pqc_d6_cz_urot}
		\end{figure*}

		\begin{table}[b]
			\ra{1.2}
			\centering
			\begin{tabular}{@{}rrcrrrcrrrcrrr@{}}
				\toprule
				\multicolumn{2}{c}{$d=12$} &
				\phantom{a} &
				\multicolumn{3}{c}{3 CNOT}&
				\phantom{a} &
				\multicolumn{3}{c}{3 CRX}&
				\phantom{a} &
				\multicolumn{3}{c}{2 CAN} \\
				\cmidrule{1-2} \cmidrule{4-6} \cmidrule{8-10} \cmidrule{12-14}
				$n$ & $L$ &&   $s$ &  $t$ &   $p$ &&  $s$ &  $t$ &   $p$ && $s$ & $t$ & $p$\\
				\midrule
				1 & 12 && 90 & 0 & 78 && 90 & 0 & 78 && 64 & 0 & 52 \\
				2 & 6 && 96 & 21 & 84 && 96 & 21 & 105 && 68 & 70 & 126 \\
				3 & 4 && 102 & 30 & 90 && 102 & 30 & 120 && 72 & 80 & 140 \\
				\B 4 & \B 3 && \B 108 & \B 36 & \B 96 && \B 108 & \B 36 & \B 132 && \B 76 & \B 88 & \B 152 \\
				6 & 2 && 120 & 45 & 108 && 120 & 45 & 153 && 84 & 102 & 174 \\
				12 & 1 && 156 & 66 & 144 && 156 & 66 & 210 && 108 & 140 & 236 \\
				\bottomrule
			\end{tabular}
			\caption{Comparison of number of single-qubit gates $s$, number of two-qubit gates $t$, and trainable parameter count $p$ for the different ans\"atze in Fig.~\ref{fig:loss_degree12_entgates}. We count one $\CAN$ gate as three two-qubit gates. The number of single-qubit gates includes the nonparameterized data-encoding gates.}
			\label{tab:parameter_count_WSW_d12}
		\end{table}
		
		\begin{table}[b]
			\ra{1.2}
			\centering
			\begin{tabular}{@{}rcrrrcrrrcrrrcrrr@{}}
				\toprule
				&
				\phantom{a} &
				\multicolumn{7}{c}{layered} &
				\phantom{a} &
				\multicolumn{7}{c}{dQNN}\\
				\cmidrule{3-9} \cmidrule{11-17}
				&&
				\multicolumn{3}{c}{CNOT} &
				\phantom{a} &
				\multicolumn{3}{c}{$CR_X$} &
				\phantom{a} &
				\multicolumn{3}{c}{$R_Y$} &
				\phantom{a} &
				\multicolumn{3}{c}{$R_YR_Z$} \\
				\cmidrule{3-5} \cmidrule{7-9} \cmidrule{11-13} \cmidrule{15-17}
				$n \times L$ &&
				s & t & p &&
				s & t & p &&
				s & t & p &&
				s & t & p \\
				\midrule
				1 &&  9 &  0 &  8 &&  9 &  0 &  8 &&  4 &   3 &  6 &&  9 &  3 &  11 \\
				2 && 18 &  4 & 16 && 18 &  4 & 20 &&  7 &   6 & 12 && 14 &  6 &  18 \\
				4 && 28 &  6 & 24 && 28 &  6 & 30 && 13 &  18 & 27 && 24 & 18 &  38 \\
				\B 6 && \B 42 & \B 12 & \B 36 && \B 42 & \B 12 & \B 48 && \B 19 & \B 30 & \B 43 && \B 34 & \B 30 &  \B 58 \\
				9 && 57 & 16 & 48 && 57 & 16 & 64 && 28 & 36 & 55 && 49 & 36 &  76 \\
				12 && 76 & 24 & 64 && 76 & 24 & 88 && 37 & 66 & 91 && 64 & 66 & 118 \\
				\bottomrule
			\end{tabular}
			\caption{Comparison of number of single-qubit gates $s$, number of two-qubit gates $t$, and trainable parameter count $p$ for the different ans\"atze in Fig.~\ref{fig_d6_compare}. We count one $\CAN$ gate in the dQNN ans\"atze as three two-qubit gates. The number of single-qubit gates includes the non-parameterized data-encoding gates.}
			\label{tab:parameter_count_WSW_dQNN}
		\end{table}
		
		\newpage
		\clearpage
		\section{Calculations}\label{ap:calculations}
		
		\subsection{Calculations for circuits representing degree one}
		
		Assume $S(x) = R_X(x) = H e^{-i \frac{x}{2} D_x} H^\dagger$ where $D_x = \text{diag}(+1, -1)$ is the diagonalized form of $\sigma_X$ diagonalized by the Hadamard matrix $H^\dagger = H = \frac{1}{\sqrt{2}}(-1)^{i \cdot j} \ket{i}\bra{j}$.
		We also assume that measurement is taken in the $z$ basis: $M = \sigma_Z = (-1)^{i\cdot j} \delta^{i}_{\ j} \ket{i} \bra{j} = (-1)^{i\cdot i} \ket{i} \bra{i}$ and that
		$W = \begin{pmatrix}
			W^{0}_{\ 0} & W^{0}_{\ 1} \\
			W^{1}_{\ 0} & W^{1}_{\ 1}
		\end{pmatrix}$ with
		$W^\dagger = \begin{pmatrix}
			\cc{W}^{\ 0}_{0} & \cc{W}^{\ 0}_{1} \\
			\cc{W}^{\ 1}_{0} & \cc{W}^{\ 1}_{1}
		\end{pmatrix}$
		is unitary.
		
		\begin{example}
			\label{ex:ws_1q1l}
			The ansatz type $W(\boldsymbol{\theta}) S(x)$ on 1 qubit and depth 1 does not learn all Fourier functions of degree 1 since the Fourier coefficient $c_0$ is equal to zero for all parameters $\boldsymbol{\theta}$.
			\begin{align*}
				WS\ket{0} &=  \frac{1}{2} (-1)^{0 \cdot j} (-1)^{j \cdot l} e^{-i \frac{x}{2}\lambda_j}  W^{q}_{ \ i} \ket{i} \\
				\bra{0} S^\dagger W^\dagger M W S \ket{0} &=
				\frac{1}{4}
				(-1)^{0 \cdot k} (-1)^{k \cdot l^\prime}
				(-1)^{0 \cdot j} (-1)^{j \cdot l}
				e^{i \frac{x}{2}(\lambda_k - \lambda_j)}
				\cc{W}^{\ l^\prime}_{i^\prime} W^{i}_{ \ l} M^{i^\prime}_{ \ i}\\
				&= \frac{1}{4}
				(-1)^{i \cdot i}
				(-1)^{k \cdot l^\prime}
				(-1)^{j \cdot l}
				e^{i \frac{x}{2}(\lambda_k - \lambda_j)}
				\cc{W}^{\ l^\prime}_{i} W^{i}_{ \ l} \\
				c_0 \colon k = j \Rightarrow c_0
				&= \frac{1}{4}
				(-1)^{i \cdot i}
				(-1)^{j \cdot l^\prime}
				(-1)^{j \cdot l}
				\cc{W}^{\ l^\prime}_{i} W^{i}_{ \ l} \\
				&= \frac{1}{4}
				(-1)^{j \cdot l^\prime}
				(-1)^{j \cdot l}
				(-1)^{i \cdot i} \cc{W}^{\ l^\prime}_{i} W^{i}_{ \ l} \\
				&= \frac{1}{4}
				(
				1 +
				(-1)^{1 \cdot l^\prime} (-1)^{1 \cdot l}
				)
				(-1)^{i \cdot i} \cc{W}^{\ l^\prime}_{i} W^{i}_{ \ l} \\
				&=
				\begin{cases}
					0, & \text{ if }\ l^\prime \neq l \\
					\frac{1}{2} (-1)^{i \cdot i} \cc{W}^{\ l}_{i} W^{i}_{ \ l}, & \text{ if }\ l = l
				\end{cases} \\
				&= \frac{1}{2}(\cc{W}^{\ l}_{0} W^{0}_{ \ l} - \cc{W}^{\ l}_{1} W^{1}_{ \ l})\\
				&= \frac{1}{2} (\cc{W}^{\ 0}_{0} W^{0}_{ \ 0} + \cc{W}^{\ 1}_{0} W^{0}_{ \ 1} - (\cc{W}^{\ 0}_{1} W^{1}_{ \ 0} + \cc{W}^{\ 1}_{1} W^{1}_{ \ 1})) \\
				& = 0
			\end{align*}
		\end{example}
		
		\begin{example}
			\label{ex:sw_1q1l}
			The ansatz $S(x)W(\boldsymbol{\theta})$ does not learn all Fourier functions of degree one since the Fourier coefficient $c_0$ is equal to zero for all parameters $\boldsymbol{\theta}$.
			\begin{align*}
				SW\ket{0} &= \frac{1}{2} (-1)^{l \cdot j} (-1)^{j \cdot q} e^{-i \frac{x}{2}\lambda_j}  W^{q}_{ \ 0} \ket{q} \\
				\bra{0}  W^\dagger S^\dagger M S W \ket{0}
				&= \frac{1}{4}
				(-1)^{i^\prime \cdot k} (-1)^{k \cdot l^\prime}
				(-1)^{i \cdot j} (-1)^{j \cdot l}
				e^{i \frac{x}{2}(\lambda_k - \lambda_j)}
				\cc{W}^{\ 0}_{l^\prime} W^{l}_{ \ 0} M^{i^\prime} { \ i}\\
				&= \frac{1}{4}
				(-1)^{i \cdot i}
				(-1)^{i \cdot k} (-1)^{k \cdot l^\prime}
				(-1)^{i \cdot j} (-1)^{j \cdot l}
				e^{i \frac{x}{2}(\lambda_k - \lambda_j)}
				\cc{W}^{\ 0}_{l^\prime} W^{l}_{ \ 0} \\
				c_0 \colon k = j \Rightarrow c_0
				&= \frac{1}{4}
				(-1)^{i \cdot i}
				(-1)^{i \cdot j} (-1)^{j \cdot l^\prime}
				(-1)^{i \cdot j} (-1)^{j \cdot l}
				\cc{W}^{\ 0}_{l^\prime} W^{l}_{ \ 0} \\
				&= \frac{1}{4}
				(-1)^{i \cdot i} (-1)^{j \cdot l^\prime} (-1)^{j \cdot l}
				\cc{W}^{\ 0}_{l^\prime} W^{l}_{ \ 0} \\
				&= \frac{1}{4}
				(1-1)(-1)^{j \cdot l^\prime} (-1)^{j \cdot l}
				\cc{W}^{\ 0}_{l^\prime} W^{l}_{ \ 0} \\
				&= 0
			\end{align*}
		\end{example}
		
		\begin{example}
			\label{ex:wsw_1q1l}
			The ansatz $U = W(\boldsymbol{\theta})S(x)W(\boldsymbol{\theta}) = R_Y(\theta) S(x) R_Y(\theta)$ does not learn all Fourier functions of degree 1 since the coefficients $c_0$ and $c_1$ are real for all parameters $\boldsymbol{\theta}$.
			\begin{align*}
				WSW\ket{0} &= \frac{1}{2} (-1)^{l_1 \cdot j} (-1)^{j \cdot l_0} e^{-i \frac{x}{2}\lambda_j}  W^{q}_{ \ l_1} W^{l_0}_{ \ 0} \ket{q} \\
				\bra{0}W^\dagger S^\dagger W^\dagger M WSW\ket{0}
				&= \frac{1}{4}
				(-1)^{i \cdot i} (-1)^{l^\prime_1 \cdot j} (-1)^{j \cdot l^\prime_0} (-1)^{l_1 \cdot j} (-1)^{j \cdot l_0}e^{-i \frac{x}{2}\lambda_j}
				\cc{W}^{\ 0}_{l_0^\prime} \cc{W}^{\ l_1^\prime}_{i} W^{i}_{ \ l_1} W^{l_0}_{ \ 0}
			\end{align*}
			If we choose $W=R_Y=\left(
			\begin{array}{cc}
				\cos(\frac{\Theta}{2}) & -\sin(\frac{\Theta}{2})\\
				\sin(\frac{\Theta}{2}) & \cos(\frac{\Theta}{2})
			\end{array}\right)$, then
			all entries $W_{ij}$ are real and, hence, the coefficients $c_0$ and $c_1$ in
			the expectation value above are real.
			Further calculations show that $c_0 = 0$ and $\abs{c_1} = 0.5$ for $W=R_X$ and $W=R_Z$.
		\end{example}
		
		\subsection{Calculations for circuits representing higher degrees}
		
		\begin{example}[one qubit, $L$ layers]
			We consider the one qubit, $L$ layer case with \\
			$S(x) = e^{-i \frac{x}{2}\sigma_x} = H e^{-i \frac{x}{2}\lambda_i} \ket{i} \bra{i} H^\dagger$ \\
			and $W(\boldsymbol{\theta}) = W^{l}_{ \ k} \ket{l}\bra{k}$
			\begin{align*}
				U(x)\ket{0} &=  W^L(\boldsymbol{\theta}) S(x) \cdots W^1(\boldsymbol{\theta})S(x) W^0(\boldsymbol{\theta})\ket{0}\\
				&= e^{-i\frac{x}{2}(\lambda_{j_1} + \cdots + \lambda_{j_L})} W^q_{ \ j_L} \cdots W^{j_2}_{ \ j_1} W^{j_1}_{ \ 0} \ket{q} \\
				&= e^{-i\frac{x}{2}\Lambda_{\boldsymbol{j}}} W^q_{ \ j_L} \cdots W^{j_1}_{ \ 0} \ket{q} \\
				\Rightarrow  \bra{0} U^\dagger(x) M U(x) \ket{0}
				&= e^{-i\frac{x}{2}(\Lambda_{\boldsymbol{j}} - \Lambda_{\boldsymbol{k}})} M^{\boldsymbol{i}^\prime}_{ \ \boldsymbol{i} } \cc{W}^{ \ 0}_{k_1} \cdots \cc{W}^{ \ k_l}_{i^\prime} W^i_{ \ j_L} \cdots W^{j_1}_{ \ 0}
			\end{align*}
			such that $\Omega = \{\frac{1}{2}(\Lambda_{\boldsymbol{k}} - \Lambda_{\boldsymbol{j}})\} = [-L, - (L-1), \cdots, -1, 0, 1, \cdots, L-1, L]$.
		\end{example}
		
		\begin{example}[$n$ qubits, one layer]
			We consider the $n$ qubit, one layer case with \\
			$S(x) = e^{-i \frac{x}{2}\sigma_x \otimes n} = H^{\otimes n} e^{-i\frac{x}{2}D_x} \otimes \ldots \otimes e^{-i\frac{x}{2}D_x} H^{\dagger \otimes n} = H^{\otimes n} e^{-i\frac{x}{2}(\lambda_{j_1} + \cdots + \lambda_{j_n})} \ket{j_1 \cdots j_n} \bra{j_1 \cdots j_n} H^{\dagger \otimes n}$ \\and  $W(\boldsymbol{\theta}) = W^{l_1 \cdots l_n}_{m_1 \cdots m_n} \ket{l_1 \cdots l_n} \bra{m_1 \cdots m_n}$
			
			\begin{align*}
				U(x)\ket{0\cdots0} &= W^1(\boldsymbol{\theta})S(x) W^0(\boldsymbol{\theta})\ket{0\cdots0} \\
				&= e^{-i\frac{x}{2}(\lambda_{j_1} + \cdots + \lambda_{j_n})} \\
				& \qquad
				W^{q_1 \cdots q_n}_{ \ j_1 \cdots j_n} W^{j_1 \cdots j_n}_{ \ 0 \cdots 0}
				\ket{q_1 \cdots q_n} \\
				&=e^{-i\frac{x}{2}\Lambda_{\boldsymbol{j}}} W^{\boldsymbol{q}}_{ \ \boldsymbol{j}}W^{\boldsymbol{j}}_{ \ \boldsymbol{0}} \ket{\boldsymbol{q}}\\
				\Rightarrow \bra{0} U^\dagger(x) M U(x) \ket{0} &= e^{-i\frac{x}{2}(\Lambda_{\boldsymbol{j}} - \Lambda_{\boldsymbol{k}})}
				M^{\boldsymbol{i}^\prime}_{ \ \boldsymbol{i} }
				\cc{W}^{ \ \boldsymbol{0}}_{\boldsymbol{k}}\cc{W}^{ \ \boldsymbol{k}}_{\boldsymbol{i}^\prime}
				W^{\boldsymbol{i}}_{ \ \boldsymbol{j}}W^{\boldsymbol{j}}_{ \ \boldsymbol{0}}
			\end{align*}
			such that $\Omega = \{\frac{1}{2}(\Lambda_{\boldsymbol{k}} - \Lambda_{\boldsymbol{j}})\} = [-n, - (n-1), \cdots, -1, 0, 1, \cdots, n-1, n]$.
		\end{example}
		
		\begin{example}[$n$ qubits, $L$ layer]
			\label{ex:nqubits_Llayer}
			We consider the $n$ qubit, $L$ layer case with\\
			$S(x) = e^{-i \frac{x}{2}\sigma_x \otimes n} = H^{\otimes n} e^{-i\frac{x}{2}D_x} \otimes \ldots \otimes e^{-i\frac{x}{2}D_x} H^{\dagger \otimes n} = H^{\otimes n} e^{-i\frac{x}{2}(\lambda_{j_1} + \cdots + \lambda_{j_n})} \ket{j_1 \cdots j_n} \bra{j_1 \cdots j_n} H^{\dagger \otimes n}$ \\and
			$W(\boldsymbol{\theta}) = W^{l_1 \cdots l_n}_{m_1 \cdots m_n} \ket{l_1 \cdots l_n} \bra{m_1 \cdots m_n}$
			
			\begin{align*}
				U(x)\ket{0\cdots0} &= W^L(\boldsymbol{\theta})S(x) \cdots W^1(\boldsymbol{\theta})S(x) W^0(\boldsymbol{\theta})\ket{0\cdots0} \\
				&= W^L(\boldsymbol{\theta})S(x) \cdots W^2(\boldsymbol{\theta})S(x) e^{-i\frac{x}{2}(\lambda_{j_{11}} + \cdots + \lambda_{j_{1n}})} W^{q_{1} \cdots q_{n}}_{ \ j_{11} \cdots j_{1n}} W^{j_{11} \cdots j_{1n}}_{ \ 0 \cdots 0} \ket{q_{1} \cdots q_{n}} \\
				&= e^{-i\frac{x}{2}(\lambda_{j_{11}} + \cdots + \lambda_{j_{1n}} + \cdots + \lambda_{j_{L1}} + \cdots + \lambda_{j_{Ln}})}
				W^{q_{1} \cdots q_{n}}_{ \ j_{L1} \cdots j_{Ln}} \cdots
				W^{j_{21} \cdots j_{2n}}_{ \ j_{11} \cdots j_{1n}}
				W^{j_{11} \cdots j_{1n}}_{ \ 0 \cdots 0}
				\ket{q_{L1} \cdots q_{Ln}} \\
				&= e^{-i\frac{x}{2}(\Lambda_{\boldsymbol{j_1}} + \cdots + \Lambda_{\boldsymbol{j_L}} )} W^{\boldsymbol{q}}_{ \ \boldsymbol{j_L}} \cdots
				W^{\boldsymbol{j_1}}_{ \ \boldsymbol{0}}
				\ket{\boldsymbol{q}} \\
				\Rightarrow \bra{0} U^\dagger(x) M U(x) \ket{0} &=
				e^{-i\frac{x}{2}(
					\Lambda_{\boldsymbol{j_1}} + \cdots + \Lambda_{\boldsymbol{j_L}}
					-
					\Lambda_{\boldsymbol{k_1}} - \cdots - \Lambda_{\boldsymbol{k_L}})}
				M^{\boldsymbol{i}^\prime}_{ \ \boldsymbol{i} }
				\cc{W}^{ \ \boldsymbol{k_L}}_{\boldsymbol{i}^\prime} \cdots
				\cc{W}^{ \ \boldsymbol{0}}_{\boldsymbol{k_1}}
				W^{\boldsymbol{i}}_{ \ \boldsymbol{j_L}} \cdots
				W^{\boldsymbol{j_1}}_{ \ \boldsymbol{0}}
			\end{align*}
			such that $\Omega = \{\frac{1}{2}(\Lambda_{\boldsymbol{k_1}} + \cdots +\Lambda_{\boldsymbol{k_L}} - \Lambda_{\boldsymbol{j_1}} + \cdots +\Lambda_{\boldsymbol{j_L}})\} = [-nL, - (nL-1), \cdots, -1, 0, 1, \cdots, nL-1, nL]$.
		\end{example}
		
		\section{Presenting selected circuits} \label{circuits}
		
		In the following we present the circuits analyzed in Fig.~\ref{fig:1q1l_sw_wsw}
		of the preliminary results in Appendix~\ref{Preliminaries} and
		in Fig~\ref{fig:coeffs_funct} in first part of the results Sec.~\ref{subsec:schuld}.
		The three circuits that lead to the results of Fig.~\ref{fig:1q1l_sw_wsw} are given in
		Fig.~\ref{fig:sw_circuit_q1l1}, Fig.~\ref{fig:wsw_circuit_q1l1_ry},
		and Fig.~\ref{fig:wsw_circuit_q1l1_ryrz}.
		The three circuits that lead to the results of Fig.~\ref{fig:coeffs_funct} are given in
		Fig.~\ref{fig:wsw_circuit_q3l2_cz}, Fig.~\ref{fig:wsw_circuit_q3l2_cx},
		Fig.~\ref{fig:wsw_circuit_q3l2_crx_el1}
		and Fig.~\ref{fig:wsw_circuit_q3l2_crx_el2}.
		
		\begin{figure}[th]
			\centering
			\begin{quantikz}
				\lstick{$\ket{0}$}  & \gate{R_X(x)} & \gate{R_Y(\theta_1)} & \gate{R_Z(\theta_2)} & \gate{R_Y(\theta_3)} & \meter{}
			\end{quantikz}
			\caption{$SW$ circuit that yields the results in the first row of Fig.~\ref{fig:1q1l_sw_wsw}.}
			\label{fig:sw_circuit_q1l1}
		\end{figure}
		
		\begin{figure}[th]
			\centering
			\begin{quantikz}
				\lstick{$\ket{0}$} & \gate{R_Y(\theta_1)} & \gate{R_X(x)} & \gate{R_Y(\theta_2)} & \meter{}
			\end{quantikz}
			\caption{$WSW$ circuit that yields the results in the second row of Fig.~\ref{fig:1q1l_sw_wsw}.}
			\label{fig:wsw_circuit_q1l1_ry}
		\end{figure}
		
		\begin{figure}[th]
			\centering
			\begin{quantikz}
				\lstick{$\ket{0}$} & \gate{R_Y(\theta_1)} & \gate{R_Z(\theta_2)} & \gate{R_X(x)} & \gate{R_Y(\theta_3)} & \gate{R_Z(\theta_4)} & \meter{}
			\end{quantikz}
			\caption{$WSW$ circuit that yields the results in the third row of Fig.~\ref{fig:1q1l_sw_wsw}.}
			\label{fig:wsw_circuit_q1l1_ryrz}
		\end{figure}
		
		\begin{figure}[th]
			\centering
			\begin{quantikz}
				\lstick{$\ket{0}$} & \gate{R_Y(\theta_1)} & \ctrl{1} & \qw & \gate{R_X(x)} & \gate{R_Y(\theta_4)} & \ctrl{1} & \qw & \gate{R_X(x)} & \gate{R_Y(\theta_7)} & \ctrl{1} & \qw & \qw \\
				\lstick{$\ket{0}$} & \gate{R_Y(\theta_2)} & \ctrl{0} & \ctrl{1} & \gate{R_X(x)} & \gate{R_Y(\theta_5)} & \ctrl{0} & \ctrl{1} & \gate{R_X(x)} & \gate{R_Y(\theta_8)} & \ctrl{0} & \ctrl{1} & \qw \\
				\lstick{$\ket{0}$} & \gate{R_Y(\theta_3)} & \qw & \ctrl{0} & \gate{R_X(x)} & \gate{R_Y(\theta_6)} & \qw & \ctrl{0} & \gate{R_X(x)} & \gate{R_Y(\theta_9} & \qw & \ctrl{0} & \meter{}
			\end{quantikz}
			\caption{Circuit that yields the results in the first row of Fig.~\ref{fig:coeffs_funct}.}
			\label{fig:wsw_circuit_q3l2_cz}
		\end{figure}
		
		\begin{figure}[th]
			\centering
			\begin{quantikz}
				\lstick{$\ket{0}$} & \gate{R_Y(\theta_1)} & \ctrl{1} & \qw & \gate{R_X(x)} & \gate{R_Y(\theta_4)} & \ctrl{1} & \qw & \gate{R_X(x)} & \gate{R_Y(\theta_7)} & \ctrl{1} & \qw & \qw \\
				\lstick{$\ket{0}$} & \gate{R_Y(\theta_2)} & \targ{} & \ctrl{1} & \gate{R_X(x)} & \gate{R_Y(\theta_5)} & \targ{} & \ctrl{1} & \gate{R_X(x)} & \gate{R_Y(\theta_8)} & \targ{} & \ctrl{1} & \qw \\
				\lstick{$\ket{0}$} & \gate{R_Y(\theta_3)} & \qw & \targ{} & \gate{R_X(x)} & \gate{R_Y(\theta_6)} & \qw & \targ{} & \gate{R_X(x)} & \gate{R_Y(\theta_9)} & \qw & \targ{} & \meter{}
			\end{quantikz}
			\caption{Circuit that yields the results in the second row of Fig.~\ref{fig:coeffs_funct}.}
			\label{fig:wsw_circuit_q3l2_cx}
		\end{figure}
		
		\newpage
		
		\begin{turnpage}
			\begin{figure*}[th]
				\centering
				\begin{quantikz}[row sep=0cm]
					\lstick{$\ket{0}$} & \gate{R_Y(\theta_1)} & \ctrl{1} & \qw & \gate{R_X(x)} & \gate{R_Y(\theta_6)} & \ctrl{1} & \qw & \gate{R_X(x)} & \gate{R_Y(\theta_{11})} & \ctrl{1} & \qw & \qw \\
					\lstick{$\ket{0}$} & \gate{R_Y(\theta_2)} & \gate{R_X(\theta_4)} & \ctrl{1} & \gate{R_X(x)} & \gate{R_Y(\theta_7)} & \gate{R_X(\theta_9)} & \ctrl{1} & \gate{R_X(x)} & \gate{R_Y(\theta_{12})} & \gate{R_X(\theta_{14})} & \ctrl{1} & \qw \\
					\lstick{$\ket{0}$} & \gate{R_Y(\theta_3)} & \qw & \gate{R_X(\theta_5)} & \gate{R_X(x)} & \gate{R_Y(\theta_8)} & \qw & \gate{R_X(\theta_{10})} & \gate{R_X(x)} & \gate{R_Y(\theta_{13})} & \qw & \gate{R_X(\theta_{15})} & \meter{}
				\end{quantikz}
				\caption{Circuit that yields the results in the third row of Fig.~\ref{fig:coeffs_funct}.}
				\label{fig:wsw_circuit_q3l2_crx_el1}
				
				\vspace*{2cm}
				
				\begin{adjustbox}{height=0.6cm}
					\begin{quantikz}[row sep=0cm]
						\lstick{$\ket{0}$} & \gate{R_Y(\theta_1)} & \ctrl{1} & \qw & \gate{R_Y(\theta_6)} & \ctrl{1} & \qw & \gate{R_X(x)} & \gate{R_Y(\theta_{11})} & \ctrl{1} & \qw & \gate{R_Y(\theta_{16})} & \ctrl{1} & \qw & \gate{R_X(x)} & \gate{R_Y(\theta_{21})} & \ctrl{1} & \qw & \gate{R_Y(\theta_{26})} & \ctrl{1} & \qw & \qw \\
						\lstick{$\ket{0}$} & \gate{R_Y(\theta_2)} & \gate{R_X(\theta_4)} & \ctrl{1} & \gate{R_Y(\theta_7)} & \gate{R_X(\theta_9)} & \ctrl{1} & \gate{R_X(x)} & \gate{R_Y(\theta_{12})} & \gate{R_X(\theta_{14})} & \ctrl{1} & \gate{R_Y(\theta_{17})} & \gate{R_X(\theta_{19})} & \ctrl{1} & \gate{R_X(x)} & \gate{R_Y(\theta_{22})} & \gate{R_X(\theta_{24})} & \ctrl{1} & \gate{R_Y(\theta_{27})} & \gate{R_X(\theta_{29})} & \ctrl{1} & \qw \\
						\lstick{$\ket{0}$} & \gate{R_Y(\theta_3)} & \qw & \gate{R_X(\theta_5)} & \gate{R_Y(\theta_8)} & \qw & \gate{R_X(\theta_{10})} & \gate{R_X(x)} & \gate{R_Y(\theta_{13})} & \qw & \gate{R_X(\theta_{15})} & \gate{R_Y(\theta_{18})} & \qw & \gate{R_X(\theta_{20})} & \gate{R_X(x)} & \gate{R_Y(\theta_{23})} & \qw & \gate{R_X(\theta_{25})} & \gate{R_Y(\theta_{28})} & \qw & \gate{R_X(\theta_{30})} & \meter{}
					\end{quantikz}
				\end{adjustbox}
				\caption{Circuit that yields the results in the fourth row of Fig.~\ref{fig:coeffs_funct}.}
				\label{fig:wsw_circuit_q3l2_crx_el2}
			\end{figure*}
		\end{turnpage}
		
		\newpage
		\clearpage
		\twocolumngrid
		
		\bibliography{main}
		
	\end{document}